\documentclass[12pt]{article}

\usepackage{amsmath,amsthm,amssymb,amscd,epsfig,url}
\usepackage{fullpage}
\usepackage{times}
\newtheorem{theorem}{Theorem}[section]
\newtheorem{lemma}[theorem]{Lemma}
\newtheorem{proposition}[theorem]{Proposition}
\newtheorem{corollary}[theorem]{Corollary}

\newtheorem{definition}[theorem]{Definition}
\newtheorem{example}[theorem]{Example}
\newtheorem{remark}[theorem]{Remark}
\newtheorem{claim}[theorem]{Claim}

\def\<{{\langle}} \def\>{{\rangle}}       
\title{Sum of squares lower bounds from symmetry and a good story}
\author{
Aaron Potechin \\
University of Chicago
}
\date{\today}

\begin{document}
\maketitle

\begin{abstract}
In this paper, we develop machinery which makes it much easier to prove sum of squares lower bounds when the problem is symmetric under permutations of $[1,n]$ and the unsatisfiability of our problem comes from integrality arguments, i.e. arguments that an expression must be an integer. Roughly speaking, to prove SOS lower bounds with our machinery it is sufficient to verify that the answer to the following three questions is yes:
\begin{enumerate}
\item Are there natural pseudo-expectation values for the problem?
\item Are these pseudo-expectation values rational functions of the problem parameters?
\item Are there sufficiently many values of the parameters for which these pseudo-expectation values correspond to the actual expected values over a distribution of solutions which is the uniform distribution over permutations of a single solution?
\end{enumerate}

We demonstrate our machinery on three problems, the knapsack problem analyzed by Grigoriev \cite{Gri01}, the MOD 2 principle (which says that the complete graph $K_n$ has no perfect matching when $n$ is odd), and the following Turan type problem: Minimize the number of triangles in a graph $G$ with a given edge density. For knapsack, we recover Grigoriev's lower bound exactly. For the MOD 2 principle, we tighten Grigoriev's linear degree sum of squares lower bound, making it exact. Finally, for the triangle problem, we prove a sum of squares lower bound for finding the minimum triangle density. This lower bound is completely new and gives a simple example where constant degree sum of squares methods have a constant factor error in estimating graph densities.
\end{abstract}

\thispagestyle{empty}
.\newpage

\section{Introduction}
The sum of squares hierarchy (which we call SOS for brevity), a hierarchy of semidefinite programs first indepedently investigated by Shor \cite{Sho87}, Nesterov \cite{Nes00}, Parrilo \cite{Par00}, Lasserre \cite{Las01}, and Grigoriev \cite{Gri01,Gri01b}, is an exciting frontier of algorithm design, complexity theory, and proof complexity. SOS is exciting because it provides a single unified framework which can be applied to give approximation algorithms for a wide variety of combinatorial optimization problems. Moreover, SOS is conjectured to be optimal for many of these problems. In particular, SOS captures the Goemans-Williamson algorithm for MAX-CUT \cite{GW95}, the Goemans-Linial relaxation for sparsest cut (analyzed by Arora, Rao, and Vazirani \cite{ARV09}), and the subexponential time algorithm for unique games found by Arora, Barak, and Steurer \cite{ABS10}. More recently, SOS has been applied directly to give algorithms for several problems including planted sparse vector \cite{BKS14}, dictionary learning \cite{BKS15}, tensor decomposition \cite{GM15, HSSS16, MSS16}, tensor completion \cite{BM16, PS17}, and quantum separability \cite{BKS17}.

That said, there are limits to the power of SOS. As shown by SOS lower bounds for constraint satisfactions problems (CSPs) \cite{Gri01b,Sch08,BCK15,KMOW17} and gadget reductions \cite{Tul09}, SOS requires degree $\Omega(n)$ (and thus exponential time) 
to solve most NP-hard problems. As shown by SOS lower bounds on planted clique and other planted problems \cite{MPW15,DM15,HKPRS16,BHK+16,HKP+17}, SOS can have difficulty distinguishing between a random input and an input which is random except for a solution which has been planted inside it. Finally, as shown by Grigoriev's SOS lower bound for the knapsack problem \cite{Gri01}, SOS has difficulty capturing integrality arguments, i.e. arguments which say that an expression must be an integer. 

In this paper, we further explore this last weakness of SOS. In particular, we develop machinery which makes it much easier to prove SOS lower bounds when the problem is symmetric and the unsatisfiability of our problem comes from integrality arguments. The usual process for proving SOS lower bounds involves finding pseudo-expectation values (see subsection \ref{SOSdefinitionsubsection}) and then proving that a matrix called the moment matrix is PSD (postive semidefinite), which can be quite difficult. Roughly speaking, to prove SOS lower bounds with our machinery it is sufficient to verify that the answer to the following three questions is yes:
\begin{enumerate}
\item Are there natural pseudo-expectation values for the problem?
\item Are these pseudo-expectation values rational functions of the problem parameters?
\item Are there sufficiently many values of the parameters for which these pseudo-expectation values correspond to the actual expected values over a distribution of solutions which is the uniform distribution over permutations of a single solution?
\end{enumerate}

We demonstrate our machinery on three problems, the knapsack problem itself, the MOD 2 principle (which says that the complete graph $K_n$ on $n$ vertices does not have a perfect matching when $n$ is odd), and the following Turan-type problem: Minimize the number of triangles in a graph $G$ with a given edge density.

\subsection{Equations and SOS lower bounds for knapsack, the MOD 2 principle, and a triangle problem}
To state our SOS lower bounds on knapsack, the MOD 2 principle, and the triangle problem, we must first express these problems as infeasible systems of polynomial equations. We do this because as we will discuss in subsection \ref{SOSdefinitionsubsection}, SOS gives a proof system for proving that systems of polynomial equations over $\mathbb{R}$ are infeasible. Our lower bounds show that SOS requires high degree to prove that the systems of equations corresponding to knapsack, the MOD 2 principle, and the triangle problem are infeasible.

In the knapsack problem, which is a classic NP-complete problem, we have a knapsack with a capacity $k$ and $n$ items of weights $\{w_1,\dots,w_n\}$. We are then asked if it is possible to fill the knapsack with items whose total weight is $k$. To express the knapsack problem with equations, we create variables $\{x_i: i \in [1,n]\}$ where we want that $x_i = 1$ if the ith weight is taken and $x_i = 0$ otherwise. We encode this and the claim that the knapsack is at full capacity with the following equations.
\begin{enumerate}
\item $\forall i \in [1,n], x^2_i - x_i = 0$ (every item is either taken or not taken)
\item $\sum_{i=1}^{n}{{x_i}w_i} - k = 0$ (the total weight is $k$)
\end{enumerate}
In this paper, we consider the case when all of the weights are $1$, in which case our equations are:
\begin{enumerate}
\item $\forall i, x^2_i - x_i = 0$
\item $\sum_{i=1}^{n}{x_i} - k = 0$
\end{enumerate}
These equations are clearly infeasible whenever $k \notin \mathbb{Z}$. However, as Grigoriev \cite{Gri01} showed, since SOS has difficulty capturing integrality arguments, SOS requires high degree to refute these equations. 
\begin{theorem}[Grigoriev's SOS lower bound for knapsack] \ \\
Degree $\min{\{2\lfloor\min{\{k,n-k\}}\rfloor + 3,n\}}$ SOS fails to prove that the knapsack equations are infeasile.
\end{theorem}
In this paper, we observe that Grigoriev's lower bound (which is tight) follows immediately from our machinery.

For the MOD 2 principle, we are asked whether the complete graph $K_n$ has a perfect matching. To express this problem with equations, we take a variable $x_{ij}$ for each possible edge $(i,j)$ and we want that $x_{ij} = 1$ if the edge $(i,j)$ is in our matching and $x_{ij} = 0$ otherwise. We encode this and the claim that we have a perfect matching as follows:
\begin{enumerate}
\item For all $i,j \in [1,n]$ such that $i < j$, $x^2_{ij} - x_{ij} = 0$
\item For all $i \in [1,n]$, $\sum_{j \in [1,n]:j \neq i}{x_{ij}}  - 1 = 0$ (where we take $x_{ij} = x_{ji}$ whenever $i > j$)
\end{enumerate}
These equations are infeasible whenever $n$ is odd. However, Grigoriev \cite{Gri01b} showed that SOS requires high degree to refute these equations. While Grigoriev's lower bound is shown via a reduction from the Tseitin equations and is tight up to a constant factor, in this paper we use our machinery to obtain the following tight SOS lower bound directly.
\begin{theorem}[SOS lower bound for the MOD 2 principle] \ \\
Degree $\frac{n-1}{2}$ SOS fails to prove that the equations for the MOD 2 principle are infeasible.
\end{theorem}
For the triangle problem, we want to minimize the number of triangles in a graph with edge density $\rho$. For this problem, Goodman \cite{Goo59} showed the following lower bound.
\begin{theorem}[Goodman's bound]
The minimal number of triangles in a graph $G$ with $n$ vertices and edge density $\rho$ is at least 
\[
t(n,\rho) := \binom{n}{3} - \frac{n(n-1)(1 - \rho)}{6}((1+2\rho)n - 2 - 2\rho)
\]
\end{theorem}
As we will discuss in Section \ref{triangleproblemsection}, this bound is tight if there is an integer $k$ such that 
\begin{enumerate}
\item $\frac{n}{k} - 1 = (1-\rho)(n-1)$
\item $n$ is divisible by $k$.
\end{enumerate}
If so, then we can take $G$ to have $k$ independent sets of size $\frac{n}{k}$ and have all of the edges between different independent sets, which  minimizes the number of triangles in $G$ and matches Goodman's bound. Otherwise, Goodman's bound cannot be achieved.

To express this problem using equations, we again create a variable $x_{ij}$ for each possible edge $(i,j)$ and we want $x_{ij} = 1$ if the edge $(i,j)$ is in the graph and $x_{ij} = 0$ if the edge $(i,j)$ is not in the graph. We encode this, the requirement the edge density is $\rho$, and the claim that Goodman's bound can be achieved with the following equations
\begin{enumerate}
\item For all $i,j \in [1,n]$ such that $i < j$, $x^2_{ij} - x_{ij} = 0$
\item $\sum_{i,j \in [1,n]: i<j}{x_{ij}} - \rho\binom{n}{2} = 0$
\item $\sum_{i,j,k \in [1,n]:i<j<k}{x_{ij}x_{ik}x_{jk}} - t(n,\rho) = 0$ where $t(n,\rho) = \binom{n}{3} - \frac{n(n-1)(1 - \rho)}{6}((1+2\rho)n - 2 - 2\rho)$
\end{enumerate}
Using our machinery, we show the following SOS lower bound which is completely new and was the motivation for developing our machinery.
\begin{theorem}[SOS lower bound for the triangle problem] \ \\
Letting $k$ be the number such that $\frac{n}{k} - 1 = (1 - \rho)(n-1)$, degree $\lfloor\min{\{k,\frac{n}{k}\}}\rfloor + 1$ SOS fails to refute the triangle problem equations.
\end{theorem}
\subsection{Relation to previous work on symmetry and SOS}
There is a considerable body of prior research on symmetry and SOS. Several works built on the difficulty on knapsack and/or further investigated symmetric polynomials on the variables $\{x_1,\dots,x_n\}$. Laurent \cite{Lau03} used the difficulty of knapsack to show that degree $\lceil\frac{n}{2}\rceil$ SOS is required to capture the CUT polytope of the complete graph. Bleckherman, Gouveia, and Pfeiffer \cite{BGP14} used the difficulty of knapsack to construct degree 4 polynomials which are non-negative but cannot be written as a sum of squares of low degree \emph{rational} functions. Lee, Prakash, Wolf, and Yuen \cite{LPWY16} showed that there are symmetric non-negative polynomials on the variables $\{x_1,\dots,x_n\}$ which cannot be approximated with low degree sums of squares. Kurpisz, Lepp{\"a}nen, and Mastrolilli \cite{KLM16} gave a general criterion for determining if a symmetric polynomial on $\{x_1,\dots,x_n\}$ is a sum of squares or not.

While these prior works give more precise results for symmetric problems on the variables $\{x_1,\dots,x_n\}$, they do not show how to handle problems which are symmetric under permutations of $[1,n]$ but have variables such as $\{x_{ij}: i < j\}$ which depend on $2$ or more indices. Thus, these prior works are incomparable with this work. 

Another line of research on symmetry and SOS which is more closely connected to this work uses symmetry to reduce the algorithmic complexity of implementing SOS. Gatermann and Parrilo \cite{GP04} showed how representation theory can be used to greatly reduce the search space for pseudo-expectation values, allowing SOS to be run more efficiently on symmetric problems. Recently, Raymond et. al. \cite{RSST18} combined the analysis of Gatermann and Parrilo with Razborov's flag algebras \cite{Raz07} to show that in the case of $k$-subset hypercubes, the resulting semidefinite program has size which is independent of $n$. These results are quite general and apply to all of the problems we are considering. That said, these results do not tell us how to find or verify pseudo-expectation values by hand, which is generally what is needed for SOS lower bounds.

In this paper, we show how the representation theory which allows Gatermann and Parrilo \cite{GP04} and Raymond et. al. \cite{RSST18} to dramatically reduce the size of the semidefinite programs for SOS on symmetric problems can also be used to help prove theoretical SOS lower bounds on symmetric problems. In particular, Theorem \ref{squarereductiontheorem}, which is a crucial part of our machinery, essentially follows from Corollary 2.6 of Raymond et. al. \cite{RSST18}. We obtain our lower bounds by combining this theorem with the additional assumption that the unsatisfiability of the problem we are analyzing comes from integrality arguments.
\subsection{Paper outline}
The remainder of the paper is organized as follows. In Section \ref{preliminarysection}, we give some preliminaries. In Section \ref{informalsummarysection} we informally state our main techniques and results. In Section \ref{symmetrylowerboundsection} we highlight how symmetry is useful for proving SOS lower bounds even without additional assumptions. In Section \ref{goodstorysection}, we rigorously define what stories and good stories are and show that good stories imply SOS lower bounds. In Section \ref{verifyingstoriessection}, we show a method for verifying that stories are good stories. Finally, in Section \ref{triangleproblemsection} we discuss the triangle problem and why the SOS lower bound for the triangle problem is noteworthy.
\section{Preliminaries}\label{preliminarysection}
Before we can describe our machinery, we must first give some preliminaries. We begin by describing the class of symmetric problems which our machinery can be applied to. We then define the sum of squares hierarchy and discuss some notation for the paper.
\subsection{Symmetric problems}\label{exampleproblemsubsection}
\begin{definition}
We make the following assumptions about the problem $P$ we are analyzing:
\begin{enumerate}
\item We assume that $P$ is a problem about hypergraphs $G$ with vertices $V(G) = [1,n]$ and a set of possible hyperedges $E_P$. We view the hyperedges $e \in E_P$ as subsets of $[1,n]$ which may be unordered or ordered depending on $P$. 
If all of these subsets have the same size $t \geq 1$ then we say that the problem $P$ has arity $t$.
\item We assume that $P$ has variables $\{x_{e}: e \in E_P\}$ and $P$ is a YES/NO question which is described by a set of problem equations $\{s_i(\{x_e:e \in E_P\}) = 0\}$. The answer to $P$ is YES if all of these equations can be satisfied simultaneously and NO otherwise.
\item We assume that the set $E_P$ of possible hyperedges and the set $\{s_i(\{x_e:e \in E_P\}) = 0\}$ of problem equations are both symmetric under permutations of $[1,n]$.
\end{enumerate}
If a problem $P$ satisfies all of these assumptions then we say that $P$ is a symmetric hypergraph problem. Since we only consider problems of this type, for brevity we will just say symmetric problem rather than symmetric hypergraph problem.
\end{definition}
\begin{example}
Symmetric problems $P$ of arity $1$ are YES/NO questions on the variables $\{x_1,\dots,x_n\}$ which are symmetric under permutations of $[1,n]$.
\end{example}
\begin{example}
For symmetric problems $P$ of arity $2$, $E_P$ is the set of subsets of $[1,n]$ of size $2$. If the subsets in $E_P$ are unordered then $G$ is an undirected graph and we have variables $\{x_{ij}:  i,j \in [1,n], i \neq j\}$ where we take $x_{ji} = x_{ij}$. If the subsets in $E_P$ are ordered then $G$ is a directed graph and we have distinct variables $\{x_{ij}:  i,j \in [1,n], i \neq j\}$.
\end{example}
\begin{remark}
While our machinery can handle symmetric problems of any arity, the examples we focus on all have arity $1$ or $2$. Knapsack with unit weights has arity $1$ while the MOD $2$ principle and the triangle problem have arity $2$ and are about undirected graphs.
\end{remark}
\begin{remark}
Since our machinery is based on polynomial interpolation, it is important that the symmetric problem $P$ does not have inequalities as well as equalities. If $P$ has inequalities then our machinery does not immediately give an SOS lower bound and more analysis is needed.
\end{remark}

\subsection{Index degree}
For our results, rather than considering the degree of a polynomial $f$, it is more natural to consider the largest number of indices mentioned in any one monomial of $f$. We call this the index degree of $f$.
\begin{definition}[Index degree] \ 
\begin{enumerate}
\item Given a monomial $p = \prod_{e \in E_p}{x_e}$, we define $I(p) = \{i: \exists e \in E_p: i \in e\}$ and we define the index degree of $p$ to be 
\[
indexdeg(p) = indexdeg_{[1,n]}(p) = |I(p)|
\] 
In other words, $indexdeg(p)$ is the number of indices which $p$ depends on.
\item Given a polynomial $f$, if $f = \sum_{j}{c_j}{p_j}$ is the decomposition of $f$ into monomials then we define the index degree of $f$ to be $indexdeg(f) = \max_{j}{\{indexdeg(p_j)\}}$
\end{enumerate}
\end{definition}
\begin{example}
If $p$ is the monomial $p = x_{12}x_{34}$ then $p$ has degree $2$ and index degree $4$.
\end{example}
\begin{example}
If $f = x_{12}x_{13} + x^4_{24}$ then $f$ has degree $4$ and index degree $3$.
\end{example}
We will also need an analagous definition where we only consider the indices outside of a subset $I \subseteq [1,n]$.
\begin{definition}[Index degree outside of I] Let $I \subseteq [1,n]$ be a subset of indices.
\begin{enumerate}
\item Given a monomial $p = \prod_{e \in E_p}{x_e}$, we define the index degree of $p$ on $[1,n] \setminus I$ to be 
\[
indexdeg_{[1,n] \setminus I}(p) = |I(p) \setminus I|
\] 
In other words, $indexdeg_{[1,n] \setminus I}(p)$ is the number of indices in $[1,n] \setminus I$ which $p$ depends on.
\item Given a polynomial $f$, if $f = \sum_{j}{c_j}{p_j}$ is the decomposition of $f$ into monomials then we define the index degree of $f$ on $[1,n] \setminus I$ to be $indexdeg_{[1,n] \setminus I}(f)  = \max_{j}{\{indexdeg_{[1,n] \setminus I}(p_j)\}}$
\end{enumerate}
\end{definition}
\subsection{SOS and pseudo-expectation values}\label{SOSdefinitionsubsection}
We now define SOS and pseudo-expectation values, which are used to prove SOS lower bounds. One way to describe SOS is through SOS/Positivstellensatz proofs, which are defined as follows:
\begin{definition}
Given a system of polynomial equations $\{s_i = 0\}$ over $\mathbb{R}$, an index degree $d$ SOS/Positivstellensatz proof of infeasibility is an equality of the form
\[
-1 = \sum_{i}{{f_i}s_i} + \sum_{j}{g^2_j}
\]
where 
\begin{enumerate}
\item $\forall i, indexdeg(f_i) + indexdeg(s_i) \leq d$
\item $\forall j, indexdeg(g_j) \leq \frac{d}{2}$
\end{enumerate}
\end{definition}
\begin{remark}
This is a proof of infeasibility because the terms ${f_i}s_i$ should all be $0$ by the problem equations and the terms $g^2_j$ must all be non-negative, so they can't possibly sum to $-1$ if all of the problem equations are satisfied.
\end{remark}
\begin{definition}
Index degree $d$ SOS gives the following feasibility test for whether a system of polynomial equations over $\mathbb{R}$ is feasible or not. If there is an index degree $d$ Positivstellensatz proof of infeasibiblity then index degree $d$ SOS says NO. Otherwise, index degree $d$ SOS says YES.
\end{definition}
\begin{remark}
Index degree d SOS may give false positives by failing to say NO on systems of equations which are infeasible but will never give a false negative.
\end{remark}
In this paper, we show SOS lower bounds for the infeasible systems of equations described in subsection \ref{exampleproblemsubsection} by showing that for small $d$ there is no index degree $d$ Positivstellensatz proof of infeasibility for our system of equations. This can be done with index degree $d$ pseudo-expectation values, which are defined as follows:
\begin{definition}\label{pseudoexpectationdefinition}
Given a system of polynomial equations $\{s_i = 0\}$ over $\mathbb{R}$, index degree $d$ pseudo-expectation values are a linear mapping $\tilde{E}$ from polynomials of index degree $\leq d$ to $\mathbb{R}$ which satisfies the following conditions:
\begin{enumerate}
\item $\tilde{E}[1] = 1$
\item $\forall i,f, \tilde{E}[fs_{i}] = 0$ whenever $indexdeg(f) + indexdeg(s_i) \leq d$
\item $\forall g, \tilde{E}[g^2] \geq 0$ whenever $indexdeg(g) \leq \frac{d}{2}$
\end{enumerate}
\end{definition}
\begin{proposition}
If there are index degree $d$ pseudo-expectation values $\tilde{E}$ for a system of polynomial equations $s_1 = 0$, $s_2 = 0$, etc. over $\mathbb{R}$, then there is no index degree $d$ Positivstellensatz proof of infeasibility for these equations.
\end{proposition}
\begin{proof}
Assume that we have both index degree $d$ pseudo-expectation values and an index degree $d$ Positivstellensatz proof of infeasibility. Applying the pseudo-expectation values to the Positivstellensatz proof, we get the following contradiction:
\[
-1 = \tilde{E}[-1] = \sum_{i}{\tilde{E}[{f_i}{s_i}]} + \sum_{j}{\tilde{E}[g^2_j]} \geq 0
\]
\end{proof}
\begin{remark}
Condition $3$ of definition \ref{pseudoexpectationdefinition} is equivalent to the statement that the moment matrix $M$ is PSD (positive semidefinite) where $M$ is indexed by monomials $p,q$ of index degree $\leq \frac{d}{2}$ and has entries $M_{pq} = \tilde{E}[pq]$. Proving SOS lower bounds usually involves proving that $M \succeq 0$, which can be quite difficult. In this paper we can instead analyze $\tilde{E}[g^2]$ more directly.
\end{remark}
\begin{remark}
The idea behind pseudo-expectation values is that they should mimic actual expected values over a distribution of solutions. In particular, as shown by the following proposition, if $\tilde{E}$ comes from a distribution over actual solutions then it automatically gives pseudo-expectation values. This fact is crucial for our results.
\end{remark}
\begin{proposition}
If the equations $\{s_i = 0\}$  are feasible over $\mathbb{R}$ and $\Omega$ is a probability distribution over actual solutions then the linear mapping $\tilde{E}[p] = E_{\Omega}[p]$ gives index degree $d$ pseudo-expectation values for these equations for all $d$.
\end{proposition}
\begin{proof}
Observe that:
\begin{enumerate}
\item For any $x \sim \Omega$, $1 = 1$. Thus, $\tilde{E}[1] = E_{\Omega}[1] = 1$
\item For any $x \sim \Omega$, for all $i,f$, $f(x)s_i(x) = 0$. Thus, for all $i,f$, $\tilde{E}[f{s_i}] = E_{\Omega}[fs_{i}] = 0$
\item For any $x \sim \Omega$, for all $g$, $g(x)^2 \geq 0$. Thus, for all $g$, $\tilde{E}[g^2] = E_{\Omega}[g^2] \geq 0$
\end{enumerate}
\end{proof}
\subsection{Sequences of distinct indices}
We will need the following definitions about sequences of distinct indices in $[1,n]$.
\begin{definition}[Operations on sequences] \ 
\begin{enumerate}
\item Given a sequence of distinct indices $A = (i_1,\dots,i_m)$, we define the set $I_A$ to be $I_A = \{i_1,\dots,i_m\}$. In other words, $I_A$ is just $A$ without the ordering.
\item Given two sequences of distinct indices $A = (i_1,\dots,i_{m_1})$ and $B = (i'_1,\dots,i'_{m_2})$, we say that $A \subseteq B$ if $m_1 \leq m_2$ and $\forall j \in [1,m_1], i'_j = i_j$.
\item Given two sequences of distinct indices $A = (i_1,\dots,i_{m_1})$ and $B = (i'_1,\dots,i'_{m_2})$ such that $I_A \cap I_B = \emptyset$, we define $A \cup B$ to be the sequence $A \cup B = (i_1,\dots,i_{m_1},i'_1,\dots,i'_{m_2}) $
\end{enumerate}
\end{definition}
In this paper, we will never consider sequences of indices which are not distinct, so we assume without stating it explicitly that all of our sequences contain distinct indices.
\section{Informal statement of techniques and results}\label{informalsummarysection}
In this section, we informally describe our techniques and results. We first show how to obtain pseudo-expectation values for symmetric problems based on stories for these problems. We then informally state our main result and sketch how to prove it.
\subsection{Finding pseudo-expectation values: Stories and a verifier/adversary game for SOS}\label{findingEsubsection}
In this subsection, we describe a verifier/adversary game which we use to find pseudo-expectation values and deduce SOS lower bounds. We then describe how the adversary can play this game using stories and describe the resulting pseudo-expectation values for knapsack, the MOD 2 principle, and the triangle problem.

The verifier/adversary game is as follows. The verifier queries sequences of indices $\{A_i\}$. For each sequence of indices $A = (i_1,\dots,i_m)$ the verifier queries, for each $j \in [1,m]$ and every possibility for what happens with the previous indices $(i_1,\dots,i_{j-1})$, the adversary must provide a probability distribution for what happens with the index $i_j$. Taken together, these answers give a probability distribution for all of the possibilities for what happens with the indices in $A$. From these probability distributions, we can obtain pseudo-expectation values.

The verifier wins if he/she detects one of the following flaws in the adversary's answers
\begin{enumerate}
\item The adversary gives a probability for some event which is either negative or undefined.
\item The adversary's answers do not result in well-defined pseudo-expectation values because they are inconsistent. More precisely, there exist two sequences of indices $A = (i_1,\dots,i_m)$ and $A' = (i'_1,\dots,i'_m)$ such that $A'$ and $A$ are equal as sets (i.e. $\{i'_1,\dots,i'_m\}$ is a permutation of $\{i_1,\dots,i_m\}$) and the resulting probability distributions for what happens with the indices $\{i_1,\dots,i_m\}$ do not match.
\item The adversary's answers result in pseudo-expectation values such that some problem equation $s_i = 0$ is violated i.e. $\tilde{E}[f{s_i}] \neq 0$ for some polynomial $f$.
\end{enumerate}
If the verifier is unable to find such a flaw then the adversary wins.
\begin{remark}
Roughly speaking, when we say that the adversary specifies what happens with a set of indices $I$ we mean that the adversary assigns values to all variables $x_e$ such that the indices of $e$ are contained in $I$. We make this more precise in Section \ref{goodstorysection}.
\end{remark}
The adversary often has a strategy for this game based on a story for what happens with the indices. For the problems we are analyzing, the adversary's stories are as follows:
\begin{enumerate}
\item For knapsack, the adversary's story is that we set $k$ out of the $n$ $x_i$ to be $1$ and set the rest to $0$.
\item For the MOD 2 principle, the adversary's story is that for each vertex $i$, the perfect matching contains precisely one of the edges which are incident to $i$.
\item For the triangle problem, the adversary's story is that we have $k$ independent sets of size $\frac{n}{k}$.
\end{enumerate}
\begin{remark}
The adversary's stories are not convicing to us, as we can understand integrality arguments. However, the adversary just has to fool SOS, which is poor at capturing integrality arguments.
\end{remark}

We now demonstrate how these stories naturally give probability distributions for what happens with the indices and thus give pseudo-expectation values.
\begin{example}[Knapsack]
For knapsack, if the verifier first queries vertex $i$, the adversary says that $x_i = 1$ with probability $\frac{k}{n}$ and $x_i = 0$ with probability $\frac{n-k}{n}$. Thus, according to the adversary the expected value of $x_i$ is $\frac{k}{n}$ so we take $\tilde{E}[x_i] = \frac{k}{n}$

If the verifier then queries $x_j$, if we have $x_i = 1$ then the adversary says that $x_j = 1$ with probability $\frac{k-1}{n-1}$ and $x_j = 0$ with probability $\frac{n-k}{n-1}$ as the adversary wants to set $k-1$ of the remaining $n-1$ variables to $1$. If we have $x_i = 0$ then the adversary instead says that $x_j = 1$ with probability $\frac{k}{n-1}$ and $x_j = 0$ with probability $\frac{n-k-1}{n-1}$ as the adversary wants to set $k$ of the remaining $n-1$ variables to $1$. Thus, according to the adversary the expected value of ${x_i}{x_j}$ is $\frac{k(k-1)}{n(n-1)}$ so we take $\tilde{E}[{x_i}{x_j}] = \frac{k(k-1)}{n(n-1)}$.

We can find $\tilde{E}[{x_i}{x_j}{x_k}]$ in a similar way. Let's say the verifier now queries $x_k$. Unless we have that $x_i = x_j = 1$, ${x_i}{x_j}{x_k} = 0$ so we can focus on the case where $x_i = x_j = 1$, which according to the adversary happens with probability $\frac{k(k-1)}{n(n-1)}$. When $x_i = x_j = 1$ the adversary says that $x_j = 1$ with probability $\frac{k-2}{n-2}$ and $x_j = 0$ with probability $\frac{n-k}{n-2}$ as the adversary wants to set $k-2$ of the remaining $n-2$ variables to $1$. Thus, according to the adversary the expected value of ${x_i}{x_j}{x_k}$ is $\frac{k(k-1)(k-2)}{n(n-1)(n-2)}$ so we take $\tilde{E}[{x_i}{x_j}{x_k}] = \frac{k(k-1)(k-2)}{n(n-1)(n-2)}$.

Following similar logic, we obtain that for all $I \subseteq [1,n]$ such that $|I| \leq d$, $\tilde{E}[\prod_{i \in I}{x_i}] = \frac{\binom{k}{|I|}}{\binom{n}{|I|}}$
\end{example}
\begin{example}[MOD 2 principle]
For the MOD 2 principle, if the verifier first queries $i$, the adversary gives no information because there is nothing distinguishing $i$ from other vertices. If the verifier then queries $j$, the adversary says that $x_{ij} = 1$ with probability $\frac{1}{n-1}$ and $x_{ij} = 0$ with probability $\frac{n-2}{n-1}$ because the adversary wants to match $1$ out of the remaining $n-1$ vertices with $i$. Thus, we take $\tilde{E}[x_{ij}] = \frac{1}{n-1}$

We now consider $\tilde{E}[x_{ij}x_{kl}]$ where $i,j,k,l$ are all distinct. $x_{ij}x_{kl} = 0$ unless $x_{ij} = 1$ so we can focus on the case when $x_{ij} = 1$, which according to the adversasry happens with probability $\frac{1}{n-1}$. In this case, if the verifier queries $k$, the adversary gives no additional information because there is nothing distinguishing $k$ from other vertices in $[1,n] \setminus (i,j)$. If the verifier then queries $l$, the adversary says that $x_{kl} = 1$ with probabililty $\frac{1}{n-3}$ and $x_{kl} = 0$ with probabililty $\frac{n-4}{n-3}$ because the adversary wants to match $1$ of the $n-3$ remaining vertices with $k$. Thus, we take $\tilde{E}[x_{ij}x_{kl}] = \frac{1}{(n-1)(n-3)}$

Following similar logic, we obtain that for all $E \subseteq \{(i,j):i,j \in [1,n], i<j\}$ such that $|E| \leq d$, $\tilde{E}[\prod_{(i,j) \in E}{x_{ij}}] = \frac{1}{\prod_{j=1}^{|E|}{(n-2j+1)}}$ if $E$ is a partial matching and $\tilde{E}[\prod_{(i,j) \in E}{x_{ij}}] = 0$ otherwise.
\end{example}
\begin{example}[Triangle Problem]
For the triangle problem, if the verifier first queries $i$, the adversary gives no information because there is nothing distinguishing $i$ from other vertices. If the verifier then queries $j$, the adversary says that $j$ is in the same independent set as $i$ with probability $\frac{\frac{n}{k}-1}{n-1}$ and is in a different independent set with probability $\frac{n-\frac{n}{k}}{n-1}$.

If the verifier then queries $k$, if $i,j$ are in the same independent set then the adversary says that $k$ is in the same independent set as $i,j$ with probability $\frac{\frac{n}{k}-2}{n-2}$ and is in a different independent set with probability $\frac{n-\frac{n}{k}}{n-2}$. If $i,j$ are in different independent sets then the adversary says that $k$ is in the same independent set as $i$ with probability $\frac{\frac{n}{k} - 1}{n-2}$, $k$ is in the same independent set as $j$ with probability $\frac{\frac{n}{k} - 1}{n-2}$, and $k$ is in a different independent set with probability $\frac{n-2\frac{n}{k}}{n-2}$. Thus, the adversary gives the following probabilities for what happens with $i,j,k$:
\begin{enumerate}
\item The probability that $i,j,k$ are all in the same independent set is $\frac{(\frac{n}{k} - 1)(\frac{n}{k} - 2)}{(n-1)(n-2)}$
\item The probability that $i,j$ are in the same independent set and $k$ is in a different independent set is $\frac{(\frac{n}{k} - 1)(n - \frac{n}{k})}{(n-1)(n-2)}$. This is also the probability that $i,k$ are in the same independent set and $j$ is in a different independent set and the probability that $j,k$ are in the same independent set and $i$ is in a different independent set.
\item The probability that $i,j,k$ are all in different independent sets is $\frac{(n - \frac{n}{k})(n - 2\frac{n}{k})}{(n-1)(n-2)}$
\end{enumerate}
This gives the following pseudo-expectation values:
\begin{enumerate}
\item $\tilde{E}[x_{ij}] = \frac{n-\frac{n}{k}}{n-1}$
\item $\tilde{E}[x_{ij}x_{ik}] = \tilde{E}[x_{ij}x_{jk}] = \tilde{E}[x_{ik}x_{jk}] = \frac{(\frac{n}{k} - 1)(n - \frac{n}{k})}{(n-1)(n-2)} + \frac{(n - \frac{n}{k})(n - 2\frac{n}{k})}{(n-1)(n-2)} = \frac{(n - \frac{n}{k})(n - \frac{n}{k}-1)}{(n-1)(n-2)}$
\item $\tilde{E}[x_{ij}x_{ik}x_{jk}] = \frac{(n - \frac{n}{k})(n - 2\frac{n}{k})}{(n-1)(n-2)}$
\end{enumerate}
\end{example}
\begin{remark}
For the triangle problem, it is difficult to write down the general expression for $\tilde{E}$ explicitly. Fortunately, as we will show, we can verify the conditions of Definition \ref{pseudoexpectationdefinition} based on the story for $\tilde{E}$
\end{remark}
\subsection{Informal statement and proof of main result}
Roughly speaking, our main result is as follows:
\begin{theorem}[Informal statement of Theorem \ref{verifyinggoodstorytheorem}]\label{informalmaintheorem}
If $P$ is a symmetric problem described by a set of problem parameters and $S$ is a story for $P$ such that
\begin{enumerate}
\item $S$ is symmetric under permutations of $[1,n]$
\item Whenever the verifier queries a sequence of indices $A$ of length at most $r$, $S$ gives a probability distribution for what happens with the indices in $A$.
\item Whenever the verifier queries a sequence of indices $A$ of length at most $n'$, $S$ gives values for the probabilities which may be negative but are still well-defined expressions.
\item All probabilities given by $S$ are rational functions of the problem parameters.
\item For sufficiently many settings of the problem parameters, $S$ corresponds to taking the uniform distribution over permutations of a single solution $G_0$ for $P$. 
\end{enumerate}
then $S$ is a level $(r,n')$ good story for $P$ which implies that index degree $\min{\{2r,n'\}}$ SOS fails to refute the equations for $P$.
\end{theorem}
\begin{remark}
We will rigorously define stories and good stories and prove that good stories imply SOS lower bounds in Section \ref{goodstorysection}.
\end{remark}
\begin{remark}
Condition 5 will be made more precise in Section \ref{verifyingstoriessection}. 
\end{remark}
\begin{proof}[Proof sketch of Theorem \ref{informalmaintheorem}]
Let $\{s_i = 0\}$ be the problem equations for $P$. As described in subsection \ref{findingEsubsection}, we can obtain pseudo-expectation values $\tilde{E}$ from the story $S$. Since we obtain values for the probabilities when we query up to $n'$ indices, we can obtain a value for $\tilde{E}[p]$ whenever $p$ has index degree at most $n'$. To prove the SOS lower bound, we need to verify the following:
\begin{enumerate}
\item $\tilde{E}$ is well-defined, i.e. for any monomial $p$ of index degree $\leq n'$ we get the same values for $\tilde{E}[p]$ regardless of which order we query the indices which $p$ depends on.
\item $\tilde{E}[f{s_i}] = 0$ whenever $indexdeg(f) + indexdeg(s_i) \leq n'$
\item $\tilde{E}[g^2] \geq 0$ whenever $indexdeg(g) \leq r$
\end{enumerate}
The key observation is that whenever $S$ corresponds to taking the uniform distribution over permutations of a single solution $G_0$, all of these statements are automatically satisfied. Since this happens sufficiently often, we can use polynomial interpolation to show that the first two statements must always be satisfied.

For the third statement, we first use symmetry (see Theorem \ref{squarereductiontheorem}) to show that it is sufficient to verify that $\tilde{E}[g^2] \geq 0$ whenever $indexdeg(g) \leq r$ and $g$ is symmetric under permutations of $[1,n] \setminus I$ for some subset $I \subseteq [1,n]$ of size at most $r$. Given such a polynomial $G$, we query the indices in $I$.

Since $|I| \leq r$ and the adversary gives non-negative probabilities when we query sequences $A$ of at most $r$ indices, we can view $\tilde{E}$ as a probability distribution over pseudo-expectations values $\tilde{E}_j$ for polynomials on the indices $[1,n] \setminus I$. Moreover, we will show that $\tilde{E}_j[fg] = \tilde{E}_j[f]\tilde{E}_j[g]$ for any polynomials $f,g$ such that $f,g$ are symmetric under permutations of $[1,n] \setminus I$ and $indexdeg_{[1,n] \setminus I}(f) + indexdeg(g)_{[1,n] \setminus I} \leq n' - |I|$. Given this, 
\[
\tilde{E}[g^2] = \sum_{j}{Pr[\tilde{E}_j]\tilde{E}_j[g^2]} = \sum_{j}{Pr[\tilde{E}_j](\tilde{E}_j[g])^2} \geq 0
\]
as needed.

To show that $\tilde{E}_j[fg] = \tilde{E}_j[f]\tilde{E}_j[g]$ for any polynomials $f,g$ such that $f,g$ are symmetric under permutations of $[1,n] \setminus I$ and $indexdeg_{[1,n] \setminus I}(f) + indexdeg(g)_{[1,n] \setminus I} \leq n' - |I|$, we observe that
whenever $S$ corresponds to taking the uniform distribution over permutations of a single solution $G_0$, each $\tilde{E}_j$ will correspond to taking the uniform distribution over permutations of a single solution $G_j$ over the variables $[1,n] \setminus I$. This implies that if $f,g$ are symmetric under permutations of $[1,n] \setminus I$, $\tilde{E}_j[fg] = \tilde{E}_j[f]\tilde{E}_j[g]$. We can then show using polynomial interpolation that we must always have $\tilde{E}_j[fg] = \tilde{E}_j[f]\tilde{E}_j[g]$.
\end{proof}
Sections \ref{goodstorysection} and \ref{verifyingstoriessection} are devoted to making this argument precise.
\begin{corollary} \ 
\begin{enumerate}
\item For all positive integers $n$ and all $k \in [0,n]$ such that $k \notin \mathbb{Z}$, index degree $\min\{2\lfloor\min{\{k,n-k\}}\rfloor + 2,n\}$ SOS fails to refute the equations for knapsack with unit weights.
\item For all positive odd integers $n$, index degree $n$ SOS fails to refute the equations for the MOD 2 principle.
\item For all positive integers $n \geq 6$ and all $k \in [1,n]$ such that $k \notin \mathbb{Z}$ or $\frac{n}{k} \notin \mathbb{Z}$, index degree $2\lfloor\min{\{k,\frac{n}{k}\}}\rfloor + 2$ SOS fails to refute the claim that Goodman's bound can be achieved for the triangle problem.
\end{enumerate}
\end{corollary}
\begin{proof}[Proof sketch]
The first and fourth conditions of Theorem \ref{informalmaintheorem} are clear from the description of the stories for knapsack, the MOD 2 priniciple, and the triangle problem and how they result in pseudo-expectation values. We now observe that 
\begin{enumerate}
\item Whenever $n,k \in \mathbb{Z}$ and $0 \leq k \leq n$, the story for knapsack corresponds to taking the uniform distribution over permutations of a single solution $G_0$ which takes the first $k$ elements.
\item Whenever $n$ is even, the story for the MOD 2 principle corresponds to taking the uniform distribution over permutations of a single solution $G_0$ which takes the matching $\{(2i-1,2i): i \in [1,\frac{n}{2}]\}$)
\item Whenever $n,k$ are positive integers and $\frac{n}{k}$ is also a positive integer, the story for the triangle problem corresponds to taking the uniform distribution over permutations of a single input $G_0$ which has independent sets 
$\{[\frac{jn}{k} + 1,(j+1)\frac{n}{k}]: j \in [0,k-1]\}$ and has all edges between the independent sets.
\end{enumerate}
As we will confirm in subsection \ref{confirmingproblemstoriessubsection}, this is sufficiently often for our purposes, so the fifth condition of Theorem \ref{informalmaintheorem} holds as well. Now we just need to determine $n'$ and $r$. For all three problems, we obtain well-defined expressions for the pseudo-expectation values whenever we query $\leq n$ indices, so we may take $n' = n$.

For knapsack, as long as we have queried at most $\lfloor\min{\{k,n-k\}}\rfloor$ indices, the next index will have a non-negative probability of being taken and a non-negative probability of not being taken. Thus, we can query $\lfloor\min{\{k,n-k\}}\rfloor + 1$ indices without encountering a negative probability so we can take $r = \lfloor\min{\{k,n-k\}}\rfloor + 1$.

For the MOD 2 principle, as long as we have queried at most $\lfloor{\frac{n}{2}}\rfloor$ indices, the next index will have a non-negative probability of being matched with any of the indices which are currently unmatched and a non-negative probability of not being matched with any previous index. Thus, we can query $\lfloor{\frac{n}{2}}\rfloor + 1$ indices without encountering a negative probability so we can take $r = \lfloor{\frac{n}{2}}\rfloor + 1$.

For the triangle problem, as long as we have queried at most $\lfloor{\min\{k,\frac{n}{k}\}}\rfloor$ indices, the next index will have a non-negative probability of being in any of the current independent sets and a non-negative probability of being in a new independent set. Thus, we can query $\lfloor{\min\{k,\frac{n}{k}\}}\rfloor + 1$ indices without encountering a negative probability so we can take $r = \lfloor{\min\{k,\frac{n}{k}\}}\rfloor + 1$.
\end{proof}

\section{Symmetry and SOS lower bounds}\label{symmetrylowerboundsection}
In this section, we highlight how symmetry can help prove SOS lower bounds even without additional assumptions. In particular, we have the following theorem which essentially follows from Corollary 2.6 of \cite{RSST18}.
\begin{theorem}\label{squarereductiontheorem}
If $\tilde{E}$ is a linear map from polynomials to $\mathbb{R}$ which is symmetric with respect to permutations of $[1,n]$ then for any polynomial $g$, we can write
\[
\tilde{E}[g^2] = \sum_{I \subseteq [1,n],j:|I| \leq indexdeg(g)}{\tilde{E}[g^2_{Ij}]}
\]
where for all $I,j$,
\begin{enumerate}
\item $g_{Ij}$ is symmetric with respect to permutations of $[1,n] \setminus I$.
\item $indexdeg(g_{Ij}) \leq indexdeg(g)$
\item $\forall i \in I, \sum_{\sigma \in S_{[1,n] \setminus (I \setminus \{i\})}}{\sigma(g_{Ij})} = 0$
\end{enumerate}
\end{theorem}
Theorem \ref{squarereductiontheorem} is very useful for proving SOS lower bounds on symmetric problems because it implies that instead of checking that $\tilde{E}[g^2] \geq 0$ for all polynomials of index degree $\leq \frac{d}{2}$, it is sufficient to check polynomials which are symmetric under permutations of all but $\frac{d}{2}$ indices. However, despite its simplicity, Theorem \ref{squarereductiontheorem} is quite deep. To prove Theorem \ref{squarereductiontheorem}, we must carefully decompose $g$ and then use symmetry to analyze all of the non-square terms of $g^2$ and either eliminate them or reduce them to square terms. Fortunately, this has already been done by Corollary 2.6 of \cite{RSST18} using representation theory. We now sketch how Theorem \ref{squarereductiontheorem} follows from Corollary 2.6 of \cite{RSST18}.
\begin{proof}[Proof sketch of Theorem \ref{squarereductiontheorem} using Corollary 2.6 of \cite{RSST18}]
We must first recall some definitions.
\begin{definition}
Let $\lambda = (\lambda_1,\dots,\lambda_k)$ be a tuple of positive integers where $\lambda_1 \geq \lambda_2 \geq \dots \lambda_k$ and $\sum_{i=1}^{k}{\lambda_i} = n$. A Young tableau $\tau_{\lambda}$ of shape $\lambda$ consists of $k$ rows of boxes where the ith row has $\lambda_i$ boxes together with an assignment of the numbers [1,n] into the $n$ boxes. These rows of boxes are aarranged so that their left sides line up.

A Young tableau is a standard Young tableau if all of its rows and columns are in increasing order.
\end{definition}
\begin{definition}[Definition 2.1 of \cite{RSST18}]
If $\oplus_{\lambda}{V_{\lambda}}$ is the isotypic decomposition of the vector space of polynomials of degree $\leq d$ and $\tau_{\lambda}$ is a Young tableau of shape $\lambda$, define 
\[
W_{\tau_{\lambda}} := V_{\lambda}^{\mathcal{R}_{\tau_{\lambda}}}
\]
to be the subspace of the isotypic $V_{\lambda}$ fixed by the action of the row group $\mathcal{R}_{\tau_{\lambda}}$ (which keeps each row of $\tau_{\lambda}$ fixed but may permute the elements within each row of $\tau_{\lambda}$)
\end{definition}
Corollary 2.6 of \cite{RSST18} (rephrased slightly) says the following:
\begin{corollary}[Corollary 2.6 of \cite{RSST18}]
Suppose $p$ is a polynomial on the variables $\{x_{ij}: i,j \in [1,n], i < j\}$ such that $p$ is symmetric under permutations of $[1,n]$ and $p$ can be written as a sum of squares of polynomials of degree $\leq d$. For each partition $\lambda \vdash n$, fix a tableau $\tau_{\lambda}$ of shape $\lambda$ and choose a vector space basis $\{b^{\tau_{\lambda}}_1,\dots,b^{\tau_{\lambda}}_{m_{\lambda}}\}$ for $W_{\tau_{\lambda}}$. Then for each partition $\lambda \in \Lambda$, there exists an $m_{\lambda} \times m_{\lambda}$ PSD matrix $Q_{\lambda}$ such that 
\[
p = \sum_{\lambda \in \Lambda}{tr(Q_{\lambda}Y^{\tau_{\lambda}})}
\]
where $\Lambda := \{\lambda \vdash n: \lambda \geq_{lex} (n-2d,1^{2d})\}$ and $Y^{\tau_{\lambda}}_{ij} := sym({b^{\tau_{\lambda}}_i}{b^{\tau_{\lambda}}_j})$
\end{corollary}
Using Corollary 2.6 of \cite{RSST18}, we can prove Theorem \ref{squarereductiontheorem} as follows. Since $\tilde{E}$ is symmetric, $\tilde{E}[g^2] = \tilde{E}[sym(g^2)]$ where $sym(g^2) = \frac{1}{n!}\sum_{\sigma \in S_n}{(\sigma(g))^2}$. Since $sym(g^2)$ is symmetric and a sum of squares, by Corollary 2.6 of \cite{RSST18}, there exist PSD matrices $Q_{\lambda}$ such that 
\[
\tilde{E}[g^2] = \sum_{\lambda \in \Lambda}{\tilde{E}[tr(Q_{\lambda}Y^{\tau_{\lambda}})]}
\]
Since $\tilde{E}$ is symmetric, this implies that 
\[
\tilde{E}[g^2] = \sum_{\lambda \in \Lambda}{\tilde{E}[tr(Q_{\lambda}{Y'}^{\tau_{\lambda}})]}
\]
where ${Y'}^{\tau_{\lambda}}_{ij} := {b^{\tau_{\lambda}}_i}{b^{\tau_{\lambda}}_j}$. Now consider each $\lambda \in \Lambda$ separately and observe that since $Q_{\lambda} \succeq 0$, we can write 
$Q_{\lambda} = \sum_{j}{{q^j}{{q^j}^T}}$ for some vectors $\{q^1,\dots,q^{m_{\lambda}}\}$. Thus,
\[
tr(Q_{\lambda}Y^{\tau_{\lambda}}) = tr(\sum_{j}{{q^j}{{q^j}^T}b^{\tau_{\lambda}}{b^{\tau_{\lambda}}}^T}) = \sum_{j}{{{q^j}^T}b^{\tau_{\lambda}}{b^{\tau_{\lambda}}}^T{q^j}} = \sum_{j}{\left(\sum_{i \in m_{\lambda}}{q^{j}_{i}b^{\tau_{\lambda}}_i}\right)^2}
\]
which means we can reexpress $sym(g^2)$ as a sum of squares, each of which has the form $\left(\sum_{i=1}^{m_{\lambda}}{{c_i}b^{\tau_{\lambda}}_i}\right)^2$ for some partition $\lambda \vdash n$, tableau $\tau_{\lambda}$ of shape $\lambda$, and coefficients $\{c_i\}$

For each square $\left(\sum_{i=1}^{m_{\lambda}}{{c_i}b^{\tau_{\lambda}}_i}\right)^2$, let $I$ be the set of indices which are not in the top row of $\tau_{\lambda}$. To show the first statement of Theorem \ref{squarereductiontheorem}, observe that permuting the indices of $[1,n] \setminus I$ is just permuting the top row of $\tau_{\lambda}$. By definition, the elements of $W_{\tau_{\lambda}}$ are all invariant under such permutations, so $\sum_{i=1}^{m_{\lambda}}{{c_i}b^{\tau_{\lambda}}_i}$ is invariant under permutations of $[1,n] \setminus I$, as needed.
\begin{remark}
In the setting of Corollary 2.6 of \cite{RSST18} the variables are $\{x_{ij}: i,j \in [1,n], i < j\}$ so if $g$ has degree $d$, $g$ can have index degree $2d$ which matches the fact that $\Lambda := \{\lambda \vdash n: \lambda \geq_{lex} (n-2d,1^{2d})\}$. To prove Thorem \ref{squarereductiontheorem} as stated using Corollary 2.6 of \cite{RSST18}, Corollary 2.6 of \cite{RSST18} must be restated in terms of index degree and the proof adjusted accordingly.
\end{remark}
The second statement of Theorem \ref{squarereductiontheorem} is trivial as all of the $b^{\tau_{\lambda}}_i$ are in the vector space of polynomials of degree $\leq d$ and thus index degree $\leq 2d$.

To show the third statement of Theorem \ref{squarereductiontheorem}, we need to prove the following lemma
\begin{lemma}\label{averagingtozerolemma}
For any $\tau_{\lambda}$, letting $I$ be the set of indices which are not in the top row of $\tau_{\lambda}$, for any $i \in I$ and any $p \in W_{\tau_{\lambda}}$,
\[
\sum_{\sigma \in S_{([1,n] \setminus I) \cup \{i\}}}{\sigma(p)} = 0
\]
\end{lemma}
\begin{proof}
In fact, a stronger statement is true. For any $p \in V_{\lambda}$, if $\lambda_1$ is the length of the first row of $\lambda$ then for any $I'$ of size at most $n-\lambda_1-1$, 
\[
\sum_{\sigma \in S_{[1,n] \setminus I'}}{\sigma(p)} = 0
\]
To see this directly, observe that $V_{\lambda}$ is isomorphic to a direct sum of copies of the Specht module associated with $\lambda$, which according to Wikipedia \cite{Wiki18} is defined as follows:
\begin{definition}
Given a Young tableau $T$ of shape $\lambda$, define $\{T\}$ to be the equivalence class of all Young tableau which have the same elements as $T$ in each row (though possibly in a different order). $\{T\}$ is called a tabloid.
\end{definition}
\begin{definition}
Given a Young tableau $T$ of shape $\lambda$, 
\begin{enumerate}
\item Define $Q_T = \{\sigma \in S_n: \sigma \text{ preserves the columns of } T\}$
\item Define $E_T = \sum_{\sigma \in Q_T}{sign(\sigma)\{\sigma(T)\}}$
\end{enumerate}
The Specht module associated with $\lambda$ is $span{\{E_T: T \text{ is a Young tableau of shape } \lambda\}}$
\end{definition}
Now observe that if $I'$ has size at most $n-\lambda_1-1$, for any young tableau $T$ of shape $\lambda$, $[1,n] \setminus I'$ will contain two elements $i,j$ in a single column of $T$. Swapping $i,j$ flips the sign of $E_T$ which implies that 
\[
\sum_{\sigma \in S_{[1,n] \setminus I'}}{\sigma(E_T)} = 0
\]
Since $\{E_T: T \text{ is a Young tableau of shape } \lambda\}$ is a basis for the Specht module, this implies that for all $p \in V_{\lambda}$, $\sum_{\sigma \in S_{[1,n] \setminus I'}}{\sigma(p)} = 0$, as needed.
\end{proof}
\end{proof}
\begin{remark}
Corollary 2.6 of \cite{RSST18} does not give us an explicit expression for $\tilde{E}[g^2]$, so we can ask whether we can obtain an explicit expression for $\tilde{E}[g^2]$. It turns out that there is such an expression but it is quite complicated. For an alternative proof of Theorem \ref{squarereductiontheorem} which is explicit and combinatorial but technical, see Appendix \ref{decompositionsection}.
\end{remark}
\begin{remark}
In this analysis, we are essentially focusing on the length of the first row of $\lambda$ and ignoring the lengths of the remaining rows of $\lambda$. In particular, we make the following claim which gives an alternative explanation for why Lemma \ref{averagingtozerolemma} is true.
\begin{definition}
Define $U_r = span\{p: \exists I \subseteq [1,n]: |I| = r, \forall \sigma \in S_{[1,n] \setminus I}, \sigma(p) = p\}$ and define 
\begin{align*}
V_r = U_r/U_{r-1} = &span\{p: \exists I \subseteq [1,n]: |I| = r, \forall \sigma \in S_{[1,n] \setminus I}, \sigma(p) = p, \\
&\forall J \subseteq [1,n]: |J| \leq r-1, \sum_{\sigma \in S_{[1,n] \setminus J}}{\sigma(p) = 0}\}
\end{align*}
\end{definition}
\begin{claim}
$V_r = \oplus_{\lambda: \lambda_1 = n-r}{V_{\lambda}}$
\end{claim}
Assuming this claim, for any $\tau_{\lambda}$, letting $I$ be the set of indices which are not in the top row of $\tau_{\lambda}$, for any $p \in W_{\tau_{\lambda}} \subseteq V_{\lambda}$ and any $J$ such that $|J| < |I|$, $\sum_{\sigma \in S_{[1,n] \setminus J}}{\sigma(p) = 0}\}$. Taking $J = I \setminus \{i\}$, Lemma \ref{averagingtozerolemma} follows.

To gain intuition for why this claim is true, consider the Young module $span{\{\{T\}: T \text{ has shape } \lambda\}}$ where $\lambda$ is a hook whose first row has length $n-r$ and whose remaining rows all have length $1$. Intuitively, this Young module corresponds to $U_r$. 

Now consider which $V_{\mu}$ are captured by this Young module and thus by $U_r$. By Young's rule (see p.56-57 of the textbook ``Representation Theory: A First Course'' by Fulton and Harris \cite{FH91}), this module is isomorphic to the direct sum of $K_{\mu\lambda}$ copies of the Specht module corresponding to each $\mu$. Here $K_{\mu\lambda}$ is the Kostka number which is nonzero if and only if $\mu \geq \lambda$ i.e. $\forall j, \sum_{i=1}^{j}{\mu_i} \geq \sum_{i=1}^{j}{\lambda_i}$. When $\lambda$ is a hook whose first row has length $n-r$ and whose remaining rows all have length $1$, this is precisely the shapes $\mu$ such that $\mu_1 \geq n-r$. Thus, we expect that 
\[
U_r = \oplus_{\lambda: \lambda_1 \geq n-r}{V_\lambda}
\]
and thus
\[
V_r = U_r/U_{r-1} = \oplus_{\lambda: \lambda_1 = n-r}{V_\lambda}
\]
\end{remark}
\section{Sum of squares lower bounds from symmetry and a good story}\label{goodstorysection}
In this section, we show how strategies for the verifier/adversary game described in subsection \ref{findingEsubsection} with certain properties, which we call good stories, imply SOS lower bounds.
\subsection{Stories}
In this subsection, we rigorously define what we mean by stories. Once the definition is understood, stories are generally recognizable on sight.
\begin{definition}
Given a subset $I$ of $[1,n]$, we define $\mathcal{P}_{I}$ to be the set of all polynomials which only depend on the variables $\{x_e:e \subseteq I\}$
\end{definition}
\begin{definition}[Stories]
Let $P$ be the problem we are anaylzing and let $A = (i_1,\dots,i_m)$ be a sequence of indices. We say that a strategy $S$ for adversary is a level $n'$ story for $(P,A)$, describing what will happen with the remaining indices after we have already queried $A$, if the following is true:
\begin{enumerate}
\item $n' \leq n - |I_A|$
\item $S$ specifies what happened with the indices in $A$. More precisely, there is a linear map $\tilde{E}_{S,A}: \mathcal{P}_{I_A} \to \mathbb{R}$ corresponding to $S$
\item For all $i \in [1,n] \setminus I_A$, $S$ gives values $\{p_{ij}\}$ for the probabilities of level $n'-1$ stories $S_{ij}$ for $(P,A \cup (i))$. 
\item We have that for all $i \in [1,n] \setminus I_A$, $\sum_{j}{p_{ij}} = 1$ and $\forall f \in \mathcal{P}_{I_A},\forall j, \tilde{E}_{S,A}[f] = \tilde{E}_{S_{ij},(A \cup (i))}[f]$
\end{enumerate}
Given a level $n'$ story $S$ for $(P,A)$, for all sequences $B$ such that $A \subseteq B$, letting $i$ be the next element in $B$ after $A$, we define $\tilde{E}_{S,B} = \sum_{j}{p_{ij}\tilde{E}_{S_{ij},B}}$
\end{definition}
\begin{remark}
Note that we do not require the values $p_{ij}$ to be non-negative in this definition.
\end{remark}
\begin{remark}
For all of our examples we will have that $n' = n - |I_A|$ but we do not force this to be the case in the definition.
\end{remark}
\subsection{Useful story properties part 1}
We now define several properties our stories may have which are useful for proving SOS lower bounds. In Section \ref{verifyingstoriessection} we will describe a method for verifying these properties.

The first property we want is that our story $S$ gives the same linear map $\tilde{E}_S$ regardless of the order we query the indices.
\begin{definition}
We say that a level $n'$ story $S$ for $(P,A)$ is self-consistent if whenever $B,B'$ are sequences such that $ A \subseteq B, A \subseteq B', |I_B \setminus I_A| \leq n', |I_B \setminus I_A| \leq n'$,
\[
\forall p \in \mathcal{P}_{I_B \cap I_{B'}}, \tilde{E}_{S,B}[p] = \tilde{E}_{S,B'}[p]
\]
If $S$ is self-consistent then we define $\tilde{E}_{S}:\{f:indexdeg_{[1,n] \setminus I_A}(f) \leq n'\} \to \mathbb{R}$ to be the linear map such that for all monomials $p$ such that $indexdeg_{[1,n] \setminus I_A}(p) \leq n'$, for any sequence $B$ of length at most $n'$ such that $I_B \cap I_A = \emptyset$ and $B$ contains all indices in variables of $p$ which are not in $I_A$, $\tilde{E}_{S}[p] = \tilde{E}_{S,(A \cup B)}[p]$
\end{definition}
A second property we want is that our story sounds like we are taking the expected values over the uniform distribution of permutations of a single input graph $G_0$. To make this precise, we note a useful property such expected values have. We then define single-graph mimics to be stories/pseudo-expectation values which also have this property.
\begin{proposition}
If $\Omega$ is the trivial distribution consisting of a single graph $G_0$ then for any polynomials $f$ and $g$, 
$E_{\Omega}[fg] = E_{\Omega}[f]E_{\Omega}[g]$
\end{proposition}
\begin{proof}
$E_{\Omega}[fg] = f(G_0)g(G_0) = E_{\Omega}[f]E_{\Omega}[g]$
\end{proof}
\begin{proposition}
If $\Omega$ is the uniform distribution over all permutations of a single graph $G_0$ then for all symmetric polynomials $f$ and $g$, 
$E_{\Omega}[fg] = E_{\Omega}[f]E_{\Omega}[g]$
\end{proposition}
\begin{proof}
For any symmetric polynomial $h$ and any permutation $\sigma$, $h(\sigma(G_0)) = h(G_0)$ which implies that $E_{\Omega}[h] = h(G_0)$. Thus, we again have that $E_{\Omega}[fg] = f(G_0)g(G_0) = E_{\Omega}[f]E_{\Omega}[g]$, as needed.
\end{proof}
\begin{remark}
The property that $E[fg] = E[f]E[g]$ for all symmetric polynomials $f,g$ is useful because it immediately  implies that for all symmetric polynomials $g$, $E[g^2] = (E[g])^2 \geq 0$.
\end{remark}
We now define single graph mimics.
\begin{definition}\label{mimicdefinition}
Let $P$ be a symmetric problem with equations $\{s_i = 0\}$ and let $I$ be a subset of $[1,n]$. We say that $\tilde{E}$ is a level $n'$ single graph mimic for $P$ on $[1,n] \setminus I$ if the following conditions hold:
\begin{enumerate}
\item $\tilde{E}: \{p:indexdeg_{[1,n]\setminus I}(p) \leq n'\} \to \mathbb{R}$ is a linear map which is symmetric under permutations of $[1,n] \setminus I$
\item For all $i$ and all polynomials $f$ such that $indexdeg_{[1,n]\setminus I}(f) + indexdeg_{[1,n]\setminus I}(s_i) \leq n'$, $\tilde{E}[fs_i] = 0$
\item For all polynomials $f,g$ which are symmetric under permutations of $[1,n]\setminus I$ such that $indexdeg_{[1,n]\setminus I}(f) + indexdeg_{[1,n]\setminus I}(g) \leq n'$, $\tilde{E}[fg] = \tilde{E}[f]\tilde{E}[g]$.
\end{enumerate}
We say that $S$ is a level $n'$ single-graph mimic for $(P,A)$ if $S$ is a self-consistent level $n'$ story for $(P,A)$ and $\tilde{E}_{S}$ is a level $n'$ single-graph mimic for $P$ on $[1,n] \setminus I_A$.
\end{definition}
A third property we want is that is that our story assigns non-negative probabilities to its substories as long as we don't query too many indices. If our story and all of its substories satisfy these three properties then we call it a good story.
\begin{definition}\label{goodstorydefinition}
We say that $S$ is a level $(r,n')$ good story for $(P,A)$ if the following conditions hold:
\begin{enumerate}
\item $S$ is a level $n'$ single graph mimic for $(P,A)$.
\item If $r > 0$ then for any $i \in [1,n] \setminus I_A$, the values $p_{ij}$ are non-negative and the stories $\{S_{ij}\}$ are all level $(r-1,n'-1)$ good stories for $(P,A \cup (i))$.
\end{enumerate}
\end{definition}
\subsection{SOS lower bounds from good stories}
We now prove that good stories imply SOS lower bounds.
\begin{theorem}\label{maintheorem}
Let $P$ be a symmetric problem with equations $\{s_i = 0\}$. If we have a level $(r,n')$ good story for $P$ then index degree $d = \min{\{2r,n'\}}$ SOS fails to refute the equations for $P$.
\end{theorem}
\begin{proof}
We need two components to prove this theorem. The first component is the following theorem which shows that if we have a good story then we satisfy all of the linear constraints on $\tilde{E}$ and we have that $\tilde{E}[g^2] \geq 0$ whenever $g$ is symmetric under permutations of all but a few indices.
\begin{theorem}\label{componentone}
Let $P$ be a symmetric graph problem with constraints $\{s_i = 0\}$ (where the $\{s_i\}$ are polynomials in the input variables). If we have a level $(r,n')$ good story $S$ for $P$ then the corresponding linear map $\tilde{E}_S: \{f:indexdeg(f) \leq n'\} \to \mathbb{R}$ satisfies the following properties
\begin{enumerate}
\item $\tilde{E}_S$ is symmetric under permutations of $[1,n]$
\item If $I \subseteq [1,n]$ is a subset of indices of size at most $r$ and $g$ is a polynomial which is symmetric under permutations of $[1,n] \setminus I$ such that $indexdeg_{[1,n] \setminus I}(g) \leq \frac{n'-|I|}{2}$ then $\tilde{E}_{S}[g^2] \geq 0$
\item For all $i$ and all $f$ such that $indexdeg(f) + indexdeg(s_i) \leq n'$, $\tilde{E}_{S}[fs_{i}] = 0$
\end{enumerate}
\end{theorem}
\begin{proof}
Since $S$ is a single graph mimic and single graph mimics are symmetric with respect to permutations of $[1,n]$, the first statement follows. Similarly, the third statement follows directly from condition 2 of Definition \ref{mimicdefinition}

For the second statement, by conditions 1 and 2 of Definition \ref{goodstorydefinition}, we can express $\tilde{E}_S$ as a probability distribution $\Omega$ over level $n-|I|$ single graph mimics $\tilde{E}_j$ for $P$ on $[1,n] \setminus I$. Since $g$ is symmetric under permutations of $[1,n] \setminus {I}$, for all of the $\tilde{E}_j$, $\tilde{E}_j[g^2] = \tilde{E}_j[g]\tilde{E}_j[g] \geq 0$. We now have that $\tilde{E}_{S}[g^2] = E_{E_j \sim \Omega}\left[\tilde{E}_j[g^2]\right] \geq 0$, as needed.
\end{proof}
The second component we need is Theorem \ref{squarereductiontheorem}, which shows that it is sufficient to verify that $\tilde{E}_S[g^2] \geq 0$ whenever $g$ is symmetric with respect to permutations of all but a few indices. which is exactly what is shown by Theorem \ref{componentone}.

With these components in hand, we now prove Theorem \ref{maintheorem}. 
We need to check the following:
\begin{enumerate}
\item Whenever $indexdeg(f) + indexdeg(s_i) \leq d = \min{\{2r,n'\}}$, $\tilde{E}_{S}[fs_i] = 0$.
\item Whenever $indexdeg(g) \leq \frac{d}{2} = \min{\{r,\frac{n'}{2}\}}$, $\tilde{E}_{S}[g^2] \geq 0$
\end{enumerate}
For the first statement, note that $indexdeg(f) + indexdeg(s_i) \leq n'$, so by Theorem \ref{componentone}, $\tilde{E}_S[fs_{i}] = 0$. For the second statement, given a polynomial $g$ of index degree at most $\frac{d}{2}$, by Theorem \ref{squarereductiontheorem} we can write
\[
\tilde{E}_{S}[g^2] = \sum_{I \subseteq [1,n],j:|I|\leq indexdeg(g)}{\tilde{E}_{S}[g^2_{Ij}]}
\]
where for all $I,j$, 
\[
\forall i \in I, \sum_{\sigma \in S_{[1,n] \setminus (I \setminus \{i\})}}{\sigma(g_{Ij})} = 0
\]
We now use the following lemma to upper bound $indexdeg_{[1,n] \setminus I}(g_{Ij})$:
\begin{lemma}
If $g_{Ij}$ is symmetric with respect to permutations of $[1,n] \setminus I$ and 
\[
\forall i \in I, \sum_{\sigma \in S_{[1,n] \setminus (I \setminus \{i\})}}{\sigma(g_{Ij})} = 0
\]
then all monomials in $g_{Ij}$ depend on all of the indices in $I$
\end{lemma}
\begin{proof}
Assume that there is an $i \in I$ and some monomial $p$ which does not depend on $i$ which has a nonzero coefficient in $g_{Ij}$. By symmetry, for all permutations $\sigma$ of $[1,n] \setminus I$, the coefficient of $\sigma(p)$ is the same as the coefficient of $p$. However, these are also the coefficients of $\sigma_2(p)$ for permutations $\sigma_2$ of $[1,n] \setminus (I \setminus \{i\})$. Since $\forall i \in I, \sum_{\sigma \in S_{[1,n] \setminus (I \setminus \{i\})}}{\sigma(g_{Ij})} = 0$, all of these coefficients must be $0$, which is a contradiction.
\end{proof}
This lemma implies that for all of the $g_{Ij}$, $indexdeg_{[1,n] \setminus I}(g_{Ij}) \leq \frac{n'}{2} - |I| \leq \frac{n' - |I|}{2}$. Thus, by Theorem \ref{componentone}, $\tilde{E}_{S}[g^2_{Ij}] \geq 0$. Since this holds for all $I,j$, $\tilde{E}_{S}[g^2] \geq 0$, as needed.
\end{proof}
\section{Verifying good stories}\label{verifyingstoriessection}
In this section, we describe a method to verify that a story $S$ is a good story. For this method, we make the following assumption.
\begin{definition}
We assume that the problem equations and $S$ depend on a set of parameters and we take $\alpha_1,\dots,\alpha_m$ to be these parameters. 
\end{definition}
\begin{remark}
For knapsack and the triangle problem, we have two parameters $n$ and $k$. For the MOD 2 principle we only have the parameter $n$.
\end{remark}
\subsection{Useful story properties part 2}
We now describe two additional properties our stories may have which are useful for verifying that they are good stories. Once the definitions are understood, these properties are generally recognizable on sight.

One property $S$ usually has is that the linear maps $\tilde{E}_{S,B}$ assign values to monomials which are rational functions of the parameters $\alpha_1,\dots,\alpha_m$.
\begin{definition}
We say that a level $n'$ story $S$ for $(P,A)$ is rational if the following conditions hold
\begin{enumerate}
\item For all $B$ such that $A \subseteq B$ and $|I_B \setminus I_A| \leq n'$, for all monomials $p$ such that $I(p) \subseteq I_B$, $\tilde{E}_{S,B}[p]$ is a rational function of the parameters $\alpha_1,\dots,\alpha_m$.
\item The rational functions $\{\tilde{E}_{S,B}[p]: A \subseteq B, |I_B \setminus I_A| \leq n', I(p) \subseteq I_B\}$ have a common denominator $q_S(\alpha_1,\dots,\alpha_m)$ and the degree of the numerator is bounded by a function of $n'$ and $indexdeg(p)$.
\end{enumerate} 
\end{definition}
\begin{remark}
If $S$ is symmetric under permutations of $[1,n] \setminus I_A$ then for a given $n'$ and $deg(p)$ there are only a finite number of $\tilde{E}_{S,B}[p]$ we need to consider, so the second condition is in fact redundant. We state this condition anyways to emphasize it, as we will be using it to verify that our stories are self-consistent and single graph mimics.
\end{remark}

A second property our stories may have is that there are many settings of the parameters $\alpha_1,\dots,\alpha_m$ for which $S$ and $\tilde{E}_S$ actually correspond to probabilities and expected values of the uniform distribution over permutations of a single input $G_0$.
\begin{definition} 
Let $S$ be a story for $(P,A)$
\begin{enumerate}
\item We say that $S$ is honest for $(\alpha_1,\dots,\alpha_m)$ if $S$ corresponds to what happens if we take the uniform distribution for all permutations of an actual input graph $G_0$ over $[1,n] \setminus I_A$ and $G_0$ satisfies the equations for $P$. Note that if this is the case then $S$ is automatically a single graph mimic for $(P,A)$ for the parameter values $(\alpha_1,\dots,\alpha_m)$ and $\tilde{E}_{S}[p] = E_{\sigma \in S_{[1,n] \setminus I_A}}[p(\sigma(G_0))]$
\item We say that $S$ is $z$-honest for $(\alpha_1,\dots,\alpha_{m-1})$ if there are at least $z$ values of $\alpha_m$ such that $S$ is honest for $(\alpha_1,\dots,\alpha_m)$.
\item For all $j \in [1,m-2]$, we say that $S$ is $z$-honest for $(\alpha_1,\dots,\alpha_{j})$ if there are at least $z$ values of $\alpha_{j+1}$ such that $S$ is $z$-honest for $(\alpha_1,\dots,\alpha_{j+1})$.
\item We say that $S$ is $z$-honest if there are at least $z$ values of $\alpha_{1}$ such that $S$ is $z$-honest for $(\alpha_1)$.
\end{enumerate}
\end{definition}
The intution is that it is difficult for SOS to determine whether the parameters take one of these values for which we actually have a solution or we are in between these values.

The following lemma is very useful
\begin{lemma}\label{interpolationlemma}
Let $S$ be a story which is $z$-honest. If $p(\alpha_1,\dots,\alpha_m)$ is a polynomial such that $deg(p) < z$ and $p(\alpha_1,\dots,\alpha_m) = 0$ whenever $S$ is honest for $(\alpha_1,\dots,\alpha_m)$ then $p(\alpha_1,\dots,\alpha_m) = 0$
\end{lemma}
\begin{proof}
We prove this lemma by induction. Assume that $p(\alpha_1,\dots,\alpha_m) = 0$ whenever $S$ is $z$-honest for $\alpha_1,\dots,\alpha_j$. 

Consider $p$ as a polynomial in the variables $\alpha_{j+1},\dots,\alpha_m$. Each monomial has a coefficient which is a polynomial $c(\alpha_1,\dots,\alpha_j)$ and we must have that $c(\alpha_1,\dots,\alpha_j) = 0$ whenever $S$ is $z$-honest for $\alpha_1,\dots,\alpha_j$. We now show that all of these coefficients $c(\alpha_1,\dots,\alpha_j)$ must be $0$ whenever $S$ is $z$-honest for $\alpha_1,\dots,\alpha_{j-1}$. To see this, consider such a polynomial $c(\alpha_1,\dots,\alpha_j)$ and assume that we have $\alpha_1,\dots,\alpha_{j-1}$ such that $S$ is $z$-honest for $\alpha_1,\dots,\alpha_{j-1}$. Considering $c$ as a polynomial in $\alpha_j$, $c(\alpha_j) = 0$ whenever $S$ is $z$-honest for $\alpha_1,\dots,\alpha_j$, which by definition happens for at least $z$ values of $\alpha_j$. Since $deg(c) < z$, we must have that $c(\alpha_1,\dots,\alpha_j) = c(\alpha_j) = 0$. Thus, $p(\alpha_1,\dots,\alpha_m) = 0$ whenever $S$ is $z$-honest for $\alpha_1,\dots,\alpha_{j-1}$, as needed.
\end{proof}
\subsection{Sufficient conditions for single graph mimics}
With these definitions, we can now give sufficient conditions for showing that a story $S$ is a single graph mimic.
\begin{lemma}\label{verifyingmimicslemma}
Let $S$ be a level $n'$ story for $(P,A)$. If $S$ and the parameter values $\alpha_1,\dots,\alpha_m$ satisfy the following conditions
\begin{enumerate}
\item $S$ is rational and symmetric with respect to permutations of $[1,n] \setminus I_A$.
\item For all $z > 0$, $S$ is $z$-honest.
\item Letting $q_S(\alpha_1,\dots,\alpha_m)$ be the common denominator for $\{\tilde{E}_{S,B}[p]: A \subseteq B, |I_B \setminus I_A| \leq n', I(p) \subseteq I_B\}$, $q_S(\alpha_1,\dots,\alpha_m) \neq 0$
\end{enumerate}
then for the parameter values $\alpha_1,\dots,\alpha_m$, $S$ is a level $n'$ single graph mimic for $(P,A)$.
\end{lemma}
\begin{proof}
We need to verify the following for the given values of $\alpha_1,\dots,\alpha_m$:
\begin{enumerate}
\item $S$ is self-consistent.
\item For all $i$ and all polynomials $f$ such that $indexdeg_{[1,n] \setminus I_A}(f) + indexdeg_{[1,n] \setminus I_A}(s_i) \leq n'$, $\tilde{E}_S[fs_i] = 0$
\item For any polynomials $f,g$ such that $f,g$ are symmetric under permutations of $[1,n] \setminus I_A$ and $indexdeg_{[1,n] \setminus I_A}(f) + indexdeg_{[1,n] \setminus I_A}(g) \leq n'$, $\tilde{E}_S[fg] = \tilde{E}_S[f]\tilde{E}_S[g]$.
\end{enumerate}
We first verify that $S$ is self-consistent for the given values of $\alpha_1,\dots,\alpha_m$. Let $p$ be a monomial and let $B,B'$ be sequences of indices such that $A \subseteq B$, $A \subseteq B'$, and $I(p) \subseteq I_B \cap I_{B'}$. Since $S$ is rational, $\tilde{E}_{S,B}[p] = \frac{p_1(\alpha_1,\dots,\alpha_m)}{q(\alpha_1,\dots,\alpha_m)}$ and $\tilde{E}_{S,B'}[p] =\frac{p_2(\alpha_1,\dots,\alpha_m)}{q(\alpha_1,\dots,\alpha_m)}$ are rational functions of the parameters $\alpha_1,\dots,\alpha_m$. Now note that whenever $S$ is honest for $(\alpha_1,\dots,\alpha_m)$, $\tilde{E}_{S,B'}[p] = \tilde{E}_{S,B}[p]$ which implies that \[
p_1(\alpha_1,\dots,\alpha_m)q_S(\alpha_1,\dots,\alpha_m) = p_2(\alpha_1,\dots,\alpha_m)q_S(\alpha_1,\dots,\alpha_m)
\] 

Since $S$ is $z$-honest for all $z > 0$, by Lemma \ref{interpolationlemma} we have that $p_{1}q_S = p_{2}q_S$ as polynomials in $\alpha_1,\dots,\alpha_m$. Plugging in our actual values of $\alpha_1,\dots,\alpha_m$, $q_S(\alpha_1,\dots,\alpha_m) \neq 0$ so $p_1(\alpha_1,\dots,\alpha_m) = p_2(\alpha_1,\dots,\alpha_m)$ and thus $\tilde{E}_{S,B'}[p] = \tilde{E}_{S,B}[p]$, as needed.

We can use similar ideas to prove the second and third statements but there is a subtle point we must be careful of. A problem equations $s_i$ may be a polynomial which is symmetric in $n \setminus I_A$ rather than being a fixed polynomial. In this case, $\tilde{E}_S[s_i]$ and $\tilde{E}_S[fs_i]$ will still be rational functions in the parameters $\alpha_1,\dots,\alpha_m$. However, the equality $\tilde{E}_S[fs_i] = \frac{p_{fs_i}(\alpha_1,\dots,\alpha_m)}{q_S(\alpha_1,\dots,\alpha_m)}$ may break down if 
\[
indexdeg_{[1,n] \setminus I_A}(f) + indexdeg_{[1,n] \setminus I_A}(s_i) > n'
\]
\begin{example}
If $f = {x_1}{x_2}$ and $s_i = \sum_{i=1}^{n}{x_i} - k$ then 
\[
f{s_i} = {x^2_1}{x_2} + {x_1}{x^2_2} + {x_1}{x_2}\sum_{i \in [1,n] \setminus \{1,2\}}{{x_i}} - k{x_1}{x_2}
\]
and by symmetry
\[
\tilde{E}_{S}[f{s_i}] = \tilde{E}_{S}[{x^2_1}{x_2}] + \tilde{E}_{S}[{x_1}{x^2_2}] + (n-2)\tilde{E}_{S}[{x_1}{x_2}{x_3}] - k\tilde{E}_{S}[{x_1}{x_2}]
\]
Thus, $f{s_i}$ generally has index degree $3$ and $\tilde{E}_{S}[fs_{i}] = \frac{p_{fs_i}(\alpha_1,\dots,\alpha_m)}{q_S(\alpha_1,\dots,\alpha_m)}$ is a rational function of the parameters $\alpha_1,\dots,\alpha_m$. However, if $n' = n = 2$ then we are missing the term ${x_1}{x_2}\sum_{i \in [1,n] \setminus \{1,2\}}{{x_i}}$ from $f{s_i}$ which may break the equality $\tilde{E}_S[fs_i] = \frac{p_{fs_i}(\alpha_1,\dots,\alpha_m)}{q_S(\alpha_1,\dots,\alpha_m)}$. Note that this problem will not occur as long as 
\[
indexdeg_{[1,n] \setminus I_A}(f) + indexdeg_{[1,n] \setminus I_A}(s_i) \leq n'
\]
\end{example}
With this point in mind, for the second statement, note that since $S$ is rational and $indexdeg(f) + indexdeg(s_i) \leq n'$, we can write $\tilde{E}_S[fs_i] = \frac{p_{fs_i}(\alpha_1,\dots,\alpha_m)}{q(\alpha_1,\dots,\alpha_m)}$. Now observe that $\tilde{E}[fs_i] = 0$ whenever $\tilde{E}$ is honest for $(\alpha_1,\dots,\alpha_m)$ and thus $p_{fs_i}(\alpha_1,\dots,\alpha_m) = 0$ whenever $S$ is honest for $(\alpha_1,\dots,\alpha_m)$. Since $S$ is $z$-honest for all $z > 0$, by Lemma \ref{interpolationlemma}, $p_{fs_i}(\alpha_1,\dots,\alpha_m) = 0$ as a polynomial. Plugging in the given values of $\alpha_1,\dots,\alpha_m$, $q(\alpha_1,\dots,\alpha_m) \neq 0$ so $\tilde{E}_S[fs_i] = \frac{p_{fs_i}(\alpha_1,\dots,\alpha_m)}{q(\alpha_1,\dots,\alpha_m)} = 0$, as needed.

Similarly, for the third statement we want to view $f$, $g$, and $fg$ as polynomials which depend on $n$ rather than being fixed polynomials. Still, since $S$ is rational and $indexdeg(f) + indexdeg(g) \leq n'$, we can write $\tilde{E}_{S}[f] = \frac{p_f(\alpha_1,\dots,\alpha_m)}{q(\alpha_1,\dots,\alpha_m)}$, $\tilde{E}_{S}[g] = \frac{p_g(\alpha_1,\dots,\alpha_m)}{q(\alpha_1,\dots,\alpha_m)}$, and $\tilde{E}_{S}[fg] = \frac{p_{fg}(\alpha_1,\dots,\alpha_m)}{q(\alpha_1,\dots,\alpha_m)}$.
Now observe that $\tilde{E}_{S}[fg] = \tilde{E}_{S}[f]\tilde{E}_{S}[g]$ whenever $S$ is honest for $(\alpha_1,\dots,\alpha_m)$ and thus 
\[
{p_f}(\alpha_1,\dots,\alpha_m){p_g}(\alpha_1,\dots,\alpha_m) - q(\alpha_1,\dots,\alpha_m)p_{fg}(\alpha_1,\dots,\alpha_m) = 0
\] whenever $S$ is honest for $(\alpha_1,\dots,\alpha_m)$. Since $S$ is $z$-honest for all $z$, by Lemma \ref{interpolationlemma}, ${p_f}{p_g} - qp_{fg} = 0$ as a polynomial.
Plugging in the given parameters $\alpha_1,\dots,\alpha_m$, $q(\alpha_1,\dots,\alpha_m) \neq 0$ so 
\[
\tilde{E}_S[fg] = \frac{p_{fg}(\alpha_1,\dots,\alpha_m)}{q(\alpha_1,\dots,\alpha_m)} = \frac{p_f(\alpha_1,\dots,\alpha_m)p_g(\alpha_1,\dots,\alpha_m)}{(q(\alpha_1,\dots,\alpha_m))^2} = 
\tilde{E}_S[f]\tilde{E}_S[g]
\]
\end{proof}
\subsection{Verifying good stories}
We are now ready to give sufficient conditions for a story to be a good story.
\begin{theorem}\label{verifyinggoodstorytheorem}
If $S$ is a story for $(P,A)$ such that
\begin{enumerate}
\item $S$ is symmetric with respect to permutations of $[1,n] \setminus I_A$
\item $S$ is rational
\item For all $z > 0$, $S$ is $z$-honest.
\end{enumerate}
then for a given choice of parameters $\alpha_1,\dots,\alpha_m$, if $n'$ and $r$ are numbers such that $n' \leq n - |I_A|$ and
\begin{enumerate}
\item If we consider up to $r$ further indices, the probabilities $p_{ij}$ are always non-negative.
\item If we consider up to $n'$ further indices, we may get negative values for some $p_{ij}$ but these values are always well-defined (i.e. the denominator is nonzero).
\end{enumerate}
then $S$ is a level $(n',r)$ good story for $(P,A)$.
\end{theorem}
\begin{proof}
Since we can consider up to $n'$ further indices and get well-defined values for the $p_{ij}$, $S$ is a level $n'$ story for $(P,A)$. Now by Lemma \ref{verifyingmimicslemma}, $S$ is a level $n'$ single graph mimic for $(P,A)$.

We now prove the theorem by induction on $r$. The base case $r = 0$ is trivial. If $r > 0$ then for all $i \in [1,n] \setminus I_A$, $S$ gives non-negative values $\{p_{ij}\}$ for the probabilities of level $n'-1$ stories $S_{ij}$ for $(P,A \cup (i))$. Now note that for each of these $S_{ij}$, the values of subsequent $p_{ij}$ will always be non-negative if we consider up to $r-1$ further indices and will be well-defined if we consider up to $n'-1$ further indices. Moreover, $S_{ij}$ is symmetric with respect to permutations of $[1,n] \setminus (I_A \cup \{i\})$, rational, and is $z$-honest because $S_{ij}$ is honest for $(\alpha_1,\dots,\alpha_m)$ whenever $S$ is honest for $(\alpha_1,\dots,\alpha_m)$. Thus, by the inductive hypothesis, each $S_{ij}$ is a level $(r-1,n'-1)$ good story for $(P,A \cup (i))$ so $S$ is a level $(r,n')$ good story for $(P,A)$, as needed.
\end{proof}
\subsection{Good stories for knapsack, the MOD 2 principle, and the triangle problem}\label{confirmingproblemstoriessubsection}
In this subsection, we apply Theorem \ref{verifyinggoodstorytheorem} to verify that our stories for knapsack, the MOD 2 principle, and the triangle problem are good stories.
\begin{theorem} \ 
\begin{enumerate}
\item Saying that we take $k$ out $n$ elements is a level $(\lfloor\min{\{k,n-k\}}\rfloor + 1,n)$ good story for the knapsack problem.
\item Saying that we every vertex is incident with precisely one edge is a level $(\lfloor{\frac{n}{2}}\rfloor + 1,n)$ good story for the MOD 2 principle.
\item Saying that we have $k$ independent sets of size $\frac{n}{k}$ is a level $(\lfloor\min{\{k,\frac{n}{k}\}}\rfloor + 1,n)$ good story for the triangle problem.
\end{enumerate}
\end{theorem}
\begin{proof}
For knapsack and the triangle problem, we take $\alpha_1 = n$ and $\alpha_2 = k$. For the MOD 2 principle, we just take $\alpha_1 = n$.

Our stories are clearly rational and symmetric with respect to permutations of $[1,n]$. We now check that they are $z$-honest for all $z$.

For knapsack, note that our story is honest for $(n,k)$ whenever $k$ is an integer between $0$ and $n$. Thus, whenever $n \geq z$ there are at least $z$ values of $k$ such that our story is honest for $(n,k)$, which implies that our story is $z$-honest for $(n)$ whenever $n \geq z$. For all $z$ there are infinitely many valules of $n$ such that $n \geq z$ so our story is $z$-honest for all $z$, as needed.

For the MOD 2 principle, note that our story is honest for $(n)$ whenever $n$ is an even integer. There are infinitely many even integers so our story is $z$-honest for all $z$, as needed.

For the triangle problem, note that ourstory is honest for $(n,k)$ whenever $k$ is an integer and $n$ is divisible by $k$. Thus, whenever $n = a!$ and $a \geq z$ then there are at least $z$ values of $k$ such that our description is honest for $(n,k)$, which implies that our story is $z$-honest for $(n)$ whenever $n = a!$ and $a \geq z$. For all $z$ there are infinitely many valules of $n$ such that $n = a!$ where $a \geq z$ so our story is $z$-honest for all $z$, as needed.

All that we have to do now is to determine $n'$ and $r$.

For knapsack, when we consider polynomials of index degree at most $n'$, the common denominator will be $n(n-1)\dots(n - n' + 1)$ as we are choosing $n'$ elements one by one from $[1,n]$. This is well-defined as long as $n' \leq n$ so we may take $n' = n$. The probabilities will be non-negative up to the $(\lfloor\min{\{k,n-k\}}\rfloor + 1)$-th index we consider, so we may take $r = \lfloor\min{\{k,n-k\}}\rfloor + 1$.

For the MOD 2 principle, when we consider polynomials of index degree at most $n'$, the common denominator will be $n(n-1)\dots(n - n' + 1)$ as we are choosing $n'$ elements one by one from $[1,n]$. This is well-defined as long as $n' \leq n$ so we may take $n' = n$. The probabilities will be non-negative up to the $(\lfloor{\frac{n}{2}}\rfloor + 1)$-th index we consider, so we may take $r = \lfloor{\frac{n}{2}}\rfloor + 1$.

For the triangle problem, when we consider polynomials of index degree at most $n'$, the common denominator will be $k^{n'}n(n-1)\dots(n - n' + 1)$. The additional $k^{n'}$ factor appears because there are $\frac{n}{k}$ choices for the first element in an independent set of size $\frac{n}{k}$, $\frac{n-k}{k}$ choices for the second element, etc. Again, this is well-defined as long as $n' \leq n$ so we may take $n' = n$. The probabilities will be non-negative up to the $(\lfloor\min{\{k,\frac{n}{k}\}}\rfloor + 1)$-th index we consider, so we may take $r = \lfloor\min{\{k,\frac{n}{k}\}}\rfloor + 1$
\end{proof}
\begin{corollary} \ 
\begin{enumerate}
\item For all positive integers $n$ and all non-integer $k \in [0,n]$, index degree $\min\{2\lfloor\min{\{k,n-k\}}\rfloor + 2,n\}$ SOS fails to refute the knapsack equations.
\item For all odd $n$, index degree $n$ SOS fails to refute the equations for the MOD 2 principle.
\item For all $n \geq 6$, and all $k \in [1,n]$ such that $k \notin \mathbb{Z}$ or $\frac{n}{k} \notin \mathbb{Z}$, index degree $2\lfloor\min{\{k,\frac{n}{k}\}}\rfloor + 2$ SOS fails to refute the claim that Goodman's bound can be achieved for the triangle problem.
\end{enumerate}
\end{corollary}
\section{Discussion of the triangle problem}\label{triangleproblemsection}
In this section, we discuss the triangle problem. We first show how SOS captures Goodman's bound. While this proof is not new, it explains why having $k$ independent sets of size $\frac{n}{k}$ is optimal. We then describe the true answer to the triangle problem found by Razborov, why this particular integrality gap is noteworthy, and how the SOS lower bound generalizes when we consider larger cliques instead of triangles.

\subsection{Proof of Goodman's bound}\label{Goodmanproofsubsection}
We recall Goodman's bound \cite{Goo59} below.
\begin{theorem}\label{basictriangletheorem}
The minimal number of triangles in a graph $G$ with $n$ vertices and edge density $\rho$ is at least $t(n,\rho) := \binom{n}{3} - \frac{n(n-1)(1 - \rho)}{6}((1+2\rho)n - 2 - 2\rho)$.
\end{theorem}
\begin{proof}
For the analysis, we consider induced subgraphs of $G$ with three vertices. The following definitions are helpful.
\begin{definition} \ 
\begin{enumerate}
\item We define $N_{3,3} = \sum_{i,j,k \in [1,n]: i<j<k}{x_{ij}x_{ik}x_{jk}}$ to be the number of triangles of $G$.
\item We define $N_{3,2} = \sum_{i,j,k \in [1,n]: i<j<k}{\left(x_{ij}x_{ik}(1 - x_{jk}) + x_{ij}(1 - x_{ik})x_{jk} + (1 - x_{ij})x_{ik}x_{jk}\right)}$ to be the number of induced subgraphs of $G$ with 3 vertices which have $2$ edges.
\item We define 
\[N_{3,1} = \sum_{i,j,k \in [1,n]: i<j<k}{\left(x_{ij}(1 - x_{ik})(1 - x_{jk}) + (1 - x_{ij})x_{ik}(1 - x_{jk}) + (1 - x_{ij})(1 - x_{ik})x_{jk} \right)}
\] 
to be the number of induced subgraphs of $G$ with 3 vertices which have $1$ edge.
\item We define  $N_{3,0} = \sum_{i,j,k \in [1,n]: i<j<k}{(1 - x_{ij})(1 - x_{ik})(1 - x_{jk})}$ to be the number of induced subgraphs of $G$ with 3 vertices which have $0$ edges.
\end{enumerate}
\end{definition}
\begin{remark}
We write out the equations for these expressions to emphasize that SOS, which works with equations, captures the arguments here
\end{remark}
\begin{proposition}
\begin{align*}
&N_{3,3} + N_{3,2} + N_{3,1} + N_{3,0} = \\
&\sum_{i,j,k \in [1,n]: i<j<k}{\left(x_{ij} + (1 - x_{ij})\right)\left(x_{ik} + (1 - x_{ik})\right)\left(x_{jk} + (1 - x_{jk})\right)} = \binom{n}{3}
\end{align*}
\end{proposition}
\begin{proof}
This proposition is just saying that every induced subgraph on $3$ vertices has either $0$, $1$, $2$, or $3$ edges.
\end{proof}
\begin{lemma}
\[
\sum_{i,j \in [1,n]: i < j}{\sum_{k \notin (i,j)}{(1 - x_{ij})}} = (n-2)\sum_{i,j \in [1,n]: i < j}{(1 - x_{ij})} = (n-2)(1-\rho)\binom{n}{2} = N_{3,2} + 2N_{3,1} + 3N_{3,0}
\]
\end{lemma}
\begin{proof}
The expression on the left counts the number of times an edge is missing from an induced subgraph on $3$ vertices. This happens once for every induced subgraph on $3$ vertices which has $2$ edges, twice for every induced subgraph on $3$ vertices which has $1$ edge, and three times for every induced subgraph on $3$ vertices which has $0$ edges.
\end{proof}
\begin{lemma}
$\sum_{i=1}^{n}{\sum_{j,k: j < k, j \neq i, k \neq i}(1 - x_{ij})(1 - x_{ik})} = N_{3,1} + 3N_{3,0}$
\end{lemma}
\begin{proof}
The expression on the left counts the number of times a pair of edges with a common endpoint are both missing from an induced subgraph on $3$ vertices. This happens once for every induced subgraph on $3$ vertices which has $1$ edge and three times for every induced subgraph on $3$ vertices which has $0$ edges.
\end{proof}
\begin{corollary}\label{trianglelowerboundcorollary}
\[
N_{3,3} = \binom{n}{3} - (n-2)(1-\rho)\binom{n}{2} + \frac{2}{3}\sum_{i=1}^{n}{\sum_{j,k: j < k, j \neq i, k \neq i}(1 - x_{ij})(1 - x_{ik})} + \frac{N_{3,1}}{3}
\]
\end{corollary}
\begin{proof}
This follows immediately from the following facts which were shown above:
\begin{enumerate}
\item $\binom{n}{3} = N_{3,3} + N_{3,2} + N_{3,1} + N_{3,0}$
\item $(n-2)(1-\rho)\binom{n}{2} = N_{3,2} + 2N_{3,1} + 3N_{3,0}$
\item $\frac{2}{3}\sum_{i=1}^{n}{\sum_{j,k: j < k, j \neq i, k \neq i}(1 - x_{ij})(1 - x_{ik})} = \frac{2}{3}N_{3,1} + 2N_{3,0}$
\end{enumerate}
\end{proof}
Using Corollary \ref{trianglelowerboundcorollary}, to lower bound the number of triangles, we should lower bound 
\[\sum_{i=1}^{n}{\sum_{j,k: j < k, j \neq i, k \neq i}(1 - x_{ij})(1 - x_{ik})}
\] and $N_{3,1}$ 
For $N_{3,1}$ we take the trivial lower bound.
\begin{proposition}
$N_{3,1} \geq 0$
\end{proposition}
For $\sum_{i=1}^{n}{\sum_{j,k: j < k, j \neq i, k \neq i}(1 - x_{ij})(1 - x_{ik})}$ we use the following lemma:
\begin{lemma}
$\sum_{i=1}^{n}{\sum_{j,k: j < k, j \neq i, k \neq i}(1 - x_{ij})(1 - x_{ik})} \geq n\binom{{(1-\rho)(n-1)}}{2}$
\end{lemma}
\begin{proof}
The tight case for this lower bound is if for each vertex $i$, there are exactly $(1-\rho)(n-1)$ $j \neq i$ which are not adjacent to $i$. To see that this is tight, consider the expression
\[
\sum_{i=1}^{n}{\left(\left(\sum_{j \in [1,n] \setminus \{i\}}{(1 - x_{ij})}\right) - (1-\rho)(n-1)\right)^2}
\]
This expression must be non-negative and is the sum of the following terms:
\begin{enumerate}
\item $2\sum_{i=1}^{n}{\sum_{j,k: j < k, j \neq i, k \neq i}(1 - x_{ij})(1 - x_{ik})}$
\item $\sum_{i=1}^{n}{\sum_{j \in [1,n] \setminus \{i\}}{(1 - x_{ij})^2}} = (1-\rho)n(n-1)$
\item $-2(1-\rho)(n-1)\left(\sum_{i=1}^{n}{\sum_{j \in [1,n] \setminus \{i\}}{(1 - x_{ij})}}\right) = -2n(1 - \rho)^2(n-1)^2$
\item $\sum_{i=1}^{n}{(1-\rho)^2(n-1)^2} = n(1 - \rho)^2(n-1)^2$
\end{enumerate}
Thus, we obtain that 
\[
2\sum_{i=1}^{n}{\sum_{j,k: j < k, j \neq i, k \neq i}(1 - x_{ij})(1 - x_{ik})} \geq n(1-\rho)(n-1)((1-\rho)(n-1) - 1)
\]
as needed
\end{proof}
Putting these bounds together, Theorem \ref{basictriangletheorem} follows. We have that
\begin{align*}
N_{3,3} &= \binom{n}{3} - (n-2)(1-\rho)\binom{n}{2} + \frac{2}{3}\sum_{i=1}^{n}{\sum_{j,k: j < k, j \neq i, k \neq i}(1 - x_{ij})(1 - x_{ik})} + \frac{N_{3,1}}{3}\\
&\geq \binom{n}{3} - \frac{1}{2}n(n-1)(n-2)(1-\rho) + \frac{1}{3}n(1-\rho)(n-1)((1-\rho)(n-1) - 1) \\
&=\binom{n}{3} - \frac{n(n-1)(1 - \rho)}{6}(3(n - 2) - 2(1-\rho)(n-1) + 2) = \\
&= \binom{n}{3} - \frac{n(n-1)(1 - \rho)}{6}((1+2\rho)n - 2 - 2\rho)
\end{align*}
\end{proof}
\begin{corollary}\label{basictrianglecorollary}
For any edge density $\rho$, the limit of the ratio of the minium number of triangles to $\binom{n}{3}$ as $n$ goes to $\infty$ is at least $1 - (1 - \rho)(1 + 2\rho) = \rho(2\rho - 1)$.
\end{corollary}
We now show that the bound in Theorem \ref{basictriangletheorem} corresponds to having $k$ independent sets of size $\frac{n}{k}$ and having all remaining edges. This makes sense as such graphs are regular and have $N_{3,1} = 0$. Thus, Corollary \ref{basictrianglecorollary} is tight when $\rho$ equals one of the critical values $\{1 - \frac{1}{k}: k \in \mathbb{Z},k\geq 2\}$
\begin{corollary}\label{tightexamplecorollary}
Given an edge density $\rho$, taking $k$ be the number such that $\frac{n}{k} - 1 = (1 - \rho)(n-1)$, 
\[
N_{3,3} \geq \frac{1}{6}n\left(n - \frac{n}{k}\right)\left(n - 2\frac{n}{k}\right)
\]
\end{corollary}
\begin{proof}
\begin{align*}
\frac{1}{6}n\left(n - \frac{n}{k}\right)\left(n - 2\frac{n}{k}\right) &= \frac{1}{6}n((n-1) - (1 - \rho)(n-1))((n-2) - 2(1 - \rho)(n-1)) \\
&= \binom{n}{3} - \frac{n(n-1)(1-\rho)}{6}\left((n-2) + 2(n-1) - 2(1 - \rho)(n-1)\right) \\ 
&= \binom{n}{3} - \frac{n(n-1)(1 - \rho)}{6}((1+2\rho)n - 2 - 2\rho)
\end{align*}
\end{proof}
\subsection{The true answer to the triangle problem}
For the triangle problem, as discussed above, if $\rho$ equals one of the critical values $\{1 - \frac{1}{t}: t \in \mathbb{Z},t\geq 2\}$ then Goodman's bound $\rho(2\rho-1)$ gives the correct asymptotic answer for the minimal triangle density. If not, then there is a conflict between making the graph regular to minimize $\sum_{i=1}^{n}{\sum_{j,k: j < k, j \neq i, k \neq i}(1 - x_{ij})(1 - x_{ik})}$ and splitting the graph into independent sets so that $N_{3,1} = 0$, so the answer is no longer clear.

It was conjectured that when $\rho = 1 - \frac{1}{t}$ and $t$ is not an integer, the optimal solution is to have $\lceil{t}\rceil - 1$ independent sets of the same size and one independent set which is smaller. This conjecture is indeed asymptotically true, but it took extensive work to prove. Bollob\'as \cite{Bol75} showed that the function $f(\rho)$ which matches Goodman's bound $\rho(2\rho - 1)$ at the critical points and is piecewise linear between them is a lower bound for the triangle density. This showed that Goodman's bound is only tight at the critical points, which was already quite interesting, but it did not give a tight bound elsewhere. Lov\'asz and Simonovits \cite{LS83} proved the conjecture in intervals near the critical points, but these intervals were incredibly small. Fisher \cite{Fis89} proved the conjecture when $\frac{1}{2} \leq \rho \leq \frac{2}{3}$. This was later independently proven by Razborov \cite{Raz07} using different techniques. Finally, Razborov \cite{Raz08} proved the full conjecture by proving the following theorem:
\begin{theorem}\label{trueanswertheorem}
Given an edge density $\rho$, let $t = \lfloor{\frac{1}{1-\rho}}\rfloor$. As $n$ goes to $\infty$, the ratio of the minimal number of triangles to $\binom{n}{3}$ is
\[
\frac{(t-1)\left(t - 2\sqrt{t(t - \rho(t+1))}\right)\left(t + \sqrt{t(t - \rho(t+1))}\right)^2}{{t^2}(t+1)^2}
\]
\end{theorem}
\subsection{A constant factor integrality gap for graph densities}
Since SOS is a very powerful tool for graph density problems, several papers \cite{Raz07,Lov08,LS12} asked whether constant degree SOS can always solve these problems asymptotically. 
Hatami and Norine \cite{HN11} answered this question negatively, showing that graph density problems are undecidable in general and giving an explicit example of a non-negative function of graph densities which cannot be written as a finite sum of squares. However, their example is not particularly natural. 

We observe here that since degree $d$ SOS fails to refute Goodman's bound when $t \geq d$ and $n$ is sufficiently large, for an appropriately chosen $\rho$ degree $d$ SOS has error $\Omega(\frac{1}{d^2})$ in finding the minimum triangle density, which is
\[
\frac{(t-1)\left(t - 2\sqrt{t(t - \rho(t+1))}\right)\left(t + \sqrt{t(t - \rho(t+1))}\right)^2}{{t^2}(t+1)^2}
\]
\begin{remark}
As observed in several prior works, based on the knapsack lower bound, SOS requires linear degree to certify that the following polynomial is non-negative:
\[
(\sum_{i=1}^{n}{x_i} - \lfloor{\frac{n}{2}}\rfloor)(\sum_{i=1}^{n}{x_i} - \lfloor{\frac{n}{2}}\rfloor + 1)
\]
However, the error for this example is quite small. In particular, while $(\sum_{i=1}^{n}{x_i} - \lfloor{\frac{n}{2}}\rfloor)(\sum_{i=1}^{n}{x_i} - \lfloor{\frac{n}{2}}\rfloor + 1)$ can have value $\Theta(n^2)$, degree 2 SOS can certify that 
\[
(\sum_{i=1}^{n}{x_i} - \lfloor{\frac{n}{2}}\rfloor)(\sum_{i=1}^{n}{x_i} - \lfloor{\frac{n}{2}}\rfloor + 1) + \frac{1}{4} \geq 0
\]
so the error is only a factor of $O(\frac{1}{n^2})$
\end{remark}
\subsection{Minimizing $k$-clique density}
In this subsection, we observe that our techniques give an analogous SOS lower bound for the following generalization of the triangle problem: Given a graph $G$ on $n$ vertices with edge density $\rho$, what is the minimum number of $k$-cliques $G$ can have? 

It turns out that this problem has the same behavior as the triangle problem. Reiher \cite{Rei16} proved that taking $t$ so that $\rho = 1 - \frac{1}{t}$, it is again optimal to have $t$ independent sets of size $\frac{n}{t}$ when  is an integer and to have $\lceil{t}\rceil - 1$ independent sets of the same size and one smaller independent set when $t$ is not an integer. Since we have the same story here, namely that we have $t$ independent sets of size $\frac{n}{t}$, we can use the same pseudo-expectation values and obtain an analogous SOS lower bound.
\\
\\
\noindent Acknowledgements: The author would like to thank Sasha Razborov for suggesting the triangle problem and for helpful conversations. The author would also like to thank Johan H{\aa}stad and Annie Raymond for helpful comments on the paper. This work was supported by the Simons Collaboration for Algorithms and Geometry, the NSF under agreement No. CCF-1412958, the Knut and Alice Wallenberg Foundation, the European Research Council, and the Swedish Research Council.\\

\begin{appendix}
\section{Explicit, combinatorial proof of Theorem \ref{squarereductiontheorem}}\label{decompositionsection}
In this appendix, we give an explicit, combinatorial proof of Theorem \ref{squarereductiontheorem}, which we restate here for convenience.
\begin{theorem}
If $\tilde{E}$ is a linear map from polynomials to $\mathbb{R}$ which is symmetric with respect to permutations of $[1,n]$ then for any polynomial $g$ such that $indexdeg(g) \leq \frac{n}{2}$, we can write
\[
\tilde{E}[g^2] = \sum_{I \subseteq [1,n],j:|I|\leq indexdeg(g)}{\tilde{E}[g^2_{Ij}]}
\]
where for all $I,j$,
\begin{enumerate}
\item $g_{Ij}$ is symmetric with respect to permutations of $[1,n] \setminus I$.
\item $indexdeg(g_{Ij}) \leq indexdeg(g)$
\item $\forall i \in I, \sum_{\sigma \in S_{[1,n] \setminus (I \setminus \{i\})}}{\sigma(g_{Ij})} = 0$
\end{enumerate}
\end{theorem}
\begin{example}
If $g = x_i$ then we can decompose $g$ as $g = g_{0} + \frac{n-1}{n}g_{1i}$ where $g_0 = \frac{1}{n}\sum_{i = 1}^{n}{x_i}$ and $g_{1i} = x_i - \frac{1}{n-1}\sum_{j \in [1,n] \setminus i}{x_j}$. We now observe that $\sum_{i=1}^{n}{{g_0}{g_{1i}}} = 0$ so $\tilde{E}[{g_0}{g_{1i}}] = 0$ for any symmetric pseudo-expectation values $\tilde{E}$. Thus, for any symmetric $\tilde{E}$, 
\[
\tilde{E}[g^2] = \tilde{E}[(g_{0} + \frac{n-1}{n}g_{1i})^2] = \tilde{E}[g^2_{0} + \frac{(n-1)^2}{n^2}g^2_{1i}] \geq 0
\]
\end{example}
\begin{example}
If $g = x_{i_1} - x_{i_2}$ we can decompose $g$ as $g = \frac{n-1}{n}(g_{1{i_1}} - g_{1{i_2}})$ where $\forall i \in [1,n], g_{1i} = x_i - \frac{1}{n-1}\sum_{j \in [1,n] \setminus \{i\}}{x_j}$. If $\tilde{E}$ is symmetric, 
for all distinct $i_1,i_2$ in $[1,n]$, $\tilde{E}[g^2_{1i_1}] = \tilde{E}[g^2_{1i_2}]$ and 
\begin{align*}
\tilde{E}[g_{1{i_1}}g_{1i_2}] &= \frac{1}{n-1}\tilde{E}\left[g_{1{i_1}}\sum_{j \neq i_1}^{n}{g_{1j}}\right] = 
\frac{1}{n-1}\tilde{E}\left[g_{1{i_1}}\sum_{j=1}^{n}{g_{1j}}\right] - \frac{1}{n-1}\tilde{E}[g^2_{1{i_1}}] \\
&= - \frac{1}{n-1}\tilde{E}[g^2_{1{i_1}}]
\end{align*}
Thus, for all symmetric $\tilde{E}$,
\[
\tilde{E}[g^2] = \frac{(n-1)^2}{n^2}\tilde{E}\left[(g_{1{i_1}} - g_{1{i_2}})^2\right] = \frac{(n-1)^2}{n^2}\left(2 + \frac{2}{n-1}\right)\tilde{E}[g^2_{1{i_1}}] 
= \frac{2(n-1)}{n}\tilde{E}[g^2_{1{i_1}}] 
\]
\end{example}
To prove Theorem \ref{squarereductiontheorem}, we first need some preliminaries.
\subsection{Preliminaries}
\subsubsection{Tuples and sets of distinct indices}
Throughout the proof of Theorem \ref{squarereductiontheorem}, we use $L = (l_1,\dots,l_r)$ to denote a tuple of distinct indices. We use the letter $L$ because these indices will be labels for vertices and we use parentheses to emphasize that the ordering matters. We will never consider tuples whose elements are not distinct, so whenever we sum over tuples $L = (l_1,\dots,l_r)$, we always require that $l_1,\dots,l_r$ are all distinct.

When we want to consider the elements of $L$ without the ordering, we denote this by $I_L$.
\begin{definition}
Given a tuple $L = (l_1,\dots,l_r)$ of distinct indices, we define $I_L = \{j: \exists i: j = l_i\}$.
\end{definition}
\subsubsection{Spanning sets for polynomials of index degree $\leq d$}
To prove Theorem \ref{squarereductiontheorem}, we need to consider various ways of expressing polynomials of index degree $\leq d$. First, we have the monomial basis.
\begin{definition}
A
\end{definition}
\begin{definition}
Gven a multi-graph $G'$ with no isolated vertices, we associate the monomial $x_{G'} = \prod_{e \in E(G')}{x_e}$ to it. We define the index degree of $G'$ to be 
\[
indexdeg(G') = |\{v: v \text{ is an endpoint of some edge } e \in E(G')\}|
\]
\end{definition}
\begin{proposition}
The monomials $\{x_{G'}: indexdeg(G') \leq d\}$ are a basis for the polynomials of index degree $\leq d$.
\end{proposition}
While it is easy to express polynomials using the monomial basis, this basis does not take advantage of any symmetries which may be present. In particular, we will often want to group monomials which are the same up to permutations of the indices together. For this, we use the flags from Razborov's flag algebras \cite{Raz07}.
\begin{definition}
We define a flag $F$ to consist of the following
\begin{enumerate}
\item A tuple of distinct vertices $V_{labeled} = (v_1, \dots, v_{r_F})$.
\item A tuple of distinct unlabeled vertices $V_{free}$
\item A multi-graph $H_F$ on vertices $V_{labeled} \cup V_{free}$ with no isolated vertices.
\end{enumerate}
We define $r_F = |V_{labeled}|$ to be the asymmetry level of the flag $F$.
\end{definition}
We now define the following spanning set for polynomials of index degree at most $d$. For each flag $F$ and ordered set $L$ of distinct labels for the vertices in $V_{labeled}$, we define a polynomial $p_{F,L}$. Roughly speaking, $p_{F,L}$ corresponds to taking the sum of $x_{G'}$ over all $G'$ which match the pattern given by $F$ and $L$.
\begin{definition}
Given a flag $F$ and an ordered set $L=(l_1, \dots l_{r_F})$ of distinct labels for $V_{labeled}$, define $p_{F,L}$ to be the polynomial
\[
p_{F,L} = \sum_{\sigma: V(H_F) \to [1,n]: \sigma \text{ is injective}, \forall i \in [1,r_F] \sigma(v_i) = l_i}{x_{\sigma(H_F)}}
\]
\end{definition}
\begin{example}
If $F$ is the flag consisting of one unlabeled triangle then 
\[p_{F,\emptyset} = 6\sum_{i<j<k}{x_{ij}x_{ik}x_{jk}} = 6(\text{\# of triangles in } G)\].
\end{example}
\begin{example}
If $F$ is the flag consisting of a single edge between one labeled vertex and one unlabeled vertex, 
\[
p_{F,\{i\}} = \sum_{j \neq i}{x_{ij}} = deg(i)
\]
\end{example}
\begin{proposition}
The polynomials $\{p_{F,L}: F \text{ has asymmetry level } \leq d\}$ are a spanning set for the polynomials of index degree $\leq d$.
\end{proposition}
\begin{proof}
Observe that for any multi-graph $G'$ on $\leq d$ vertices, $x_{G'} = p_{F,L}$ where $H_F$ is the same as $G'$ up to the labeling, $V_{labeled} = V(H_F)$, and $L$ is the tuple of vertices of $G'$.
\end{proof}
The polynomials $p_{F,L}$ allow us to group terms which are the same up to permutations of the indices together. However, it turns out that these polynomials aren't quite the right polynomials to use for decomposing $\tilde{E}[g^2]$. The reason is that we would like to have that $\tilde{E}[p_{F,L}p_{F',L'}] = 0$ whenever $F$ and $F'$ have different asymmetry levels, which would allow us to consider each asymmetry level separately. However, as shown by the following proposition, this does not hold for the polynomials $\{p_{F,L}\}$
\begin{proposition}
Let $F$ be a flag of asymmetry level $r$. For all $i \in [1,r]$, 
\[
\sum_{a \in ([1,n] \setminus I_L) \cup \{l_i\}}{p_{F,(l_1,\dots,l_{i-1},a,l_{i+1},\dots,l_{r})}} = p_{F',(l_1,\dots,l_{i-1},l_{i+1},\dots,l_{r})}
\]
where $F'$ is the flag obtained by taking $F$ and making the vertex $v_i$ unlabeled.
\end{proposition}
To obtain a spanning set where polynomials of different asymmetry levels are orthogonal to each other, for each flag $F$ and ordered set of labels $L$ for the vertices in $V_{labeled}$, we define another polynomial $\phi_{F,L}$ as follows.
\begin{definition}
Given ordered tuples of indices $L = (l_1,\dots,l_r)$ and $L' = (l'_1,\dots,l'_r)$, define $c(L,L')$ to be $0$ if $\exists i,j: i \neq j, l'_j = l_i$. Otherwise, take
\[
c(L,L') = \frac{(-1)^{|\{i:l'_i \neq l_i\}|}}{\prod_{j=0}^{|\{i:l'_i \neq l_i\}|-1}{(n-r-j)}}
\]
\end{definition}
\begin{definition}
Given a flag $F$ of asymmetry level $r$ and distinct labels $L = (l_1,\dots,l_r)$, define 
\[
\phi_{F,L} = \sum_{L' = (l'_1,\dots,l'_r)}{c(L,L')p_{F,L'}}
\]
\end{definition}
\begin{example}
Whenever $F$ has asymmetry level $0$, $\phi_{F,L} = p_{F,L}$
\end{example}
\begin{example}
If $F$ is the flag consisting of a single labeled vertex $v$ and a single hyperedge $e = \{v\}$ then $p_{F,(i)} = x_i$ and $\phi_{F,(i)} = x_i - \frac{1}{n-1}\sum_{j \neq i}{x_j}$. Note that $\sum_{i \in [1,n]}{\phi_{F,(i)}} = 0$
\end{example}
\begin{example}
If $F$ is the flag consisting of two labeled vertices and an edge between them then $p_{F,(i,j)} = x_{ij}$ while 
\[
\phi_{F,(i,j)} = x_{ij} - \frac{1}{n-2}\sum_{j' \notin (i,j)}{x_{ij'}} - \frac{1}{n-2}\sum_{i' \notin (i,j)}{x_{i'j}} + \frac{1}{(n-2)(n-3)}\sum_{i',j': i' \notin {i,j}, j' \notin \{i',i,j\}}{x_{i'j'}}
\]
Note that $\sum_{j \in [1,n] \setminus \{i\}}{\phi_{F,(i,j)}} = 0$ and $\sum_{i \in [1,n] \setminus \{j\}}{\phi_{F,(i,j)}} = 0$
\end{example}
\begin{proposition} \ 
\begin{enumerate}
\item For all permutations $\sigma \in S_n$, $c(\sigma(L),\sigma(L')) = c(L,L')$
\item For all permutations $\sigma$ of $[1,n] \setminus I_L$, $p_{F,\sigma(L)} = p_{F,L}$, $\phi_{F,\sigma(L)} = \phi_{F,L}$, and 
\[
c(L,\sigma(L')) = c(\sigma(L),\sigma(L')) = c(L,L')
\]
\end{enumerate}
\end{proposition}
\subsection{Facts about $\phi_{F,L}$}
In this subsection, we prove the following useful facts about the polynomials $\phi_{F,L}$.
\begin{theorem}\label{phipropertiestheorem} \ 
\begin{enumerate}
\item For any flag $F$, tuple of indices $L = (l_1,\dots,l_{r_F})$, and any $i \in [1,r_F]$, 
\[
\sum_{a \in ([1,n] \setminus I_{L}) \cup \{l_i\})}{\phi_{F,\{l_1,\dots,l_{i-1},a,l_{i+1},\dots,l_r\}}} = 0
\]
More generally, for any flag $F$, tuple of indices $L = (l_1,\dots,l_{r_F})$, and any set of indices $I$ such that $I \subsetneq I_{L}$, 
\[
\sum_{\sigma \in S_{[1,n] \setminus I}}{\phi_{F,\sigma(L)}} = 0
\]
\item If $F_1$ and $F_2$ are flags such that $r_{F_1} \neq r_{F_2}$ then for any $\tilde{E}$ which is symmetric under permutations of $[1,n]$, for any tuples $L,L'$ of sizes $r_{F_1}$ and $r_{F_2}$, 
$\tilde{E}[\phi_{F_1,L} \cdot \phi_{F_2,L'}] = 0$
\item For all polynomials $g$ such that $indexdeg(g) \leq \frac{n}{2}$, we can write $g = \sum_{F,L}{b_{F,L}\phi_{F,L}}$ where for all flags $F$, tuples of indices $L = (l_1,\dots,l_{r_F})$, and $i \in [1,r_F]$, 
\[
\sum_{a \in ([1,n] \setminus I_L) \cup l_i}{b_{F,(l_1,\dots,l_{i-1},a,l_{i+1},\dots,l_{r_F})}} = 0
\]
\end{enumerate}
\end{theorem}
\begin{proof}
To prove this theorem, we need the following key fact about the coefficients $c(L,L')$:
\begin{lemma}\label{coefficientequationlemma}
For all tuples of indices $L,L'$ of size $r$ and all $i \in [1,r]$,
\[
\sum_{a \in ([1,n] \setminus I_{L'}) \cup \{l'_i\})}{c(L,(l'_1,\dots,l'_{i-1},a,l'_{i+1},\dots,l'_r))} = 0
\]
\end{lemma}
\begin{proof}
If there exist $j,j'$ such that $j' \neq j$, $j' \neq i$, and $l'_{j'} = l_j$ then 
\[
\forall a \in ([1,n] \setminus I_{L'}) \cup \{l'_i\}, c(L,(l'_1,\dots,l'_{i-1},a,l'_{i+1},\dots,l'_r)) = 0
\]
Otherwise, take 
$L'_2 = (l'_1,\dots,l'_{i-1},l_i,l'_{i+1},\dots,l'_r)$, let $z = |\{j \neq i: l'_j \neq l_j\}|$, and observe that
\begin{align*}
&\sum_{a \in ([1,n] \setminus I_{L'}) \cup \{l'_i\}}{c(L,(l'_1,\dots,l'_{i-1},a,l'_{i+1},\dots,l'_r))} = \\
&c(L,L'_2) + \sum_{a \in [1,n] \setminus (I_L \cup I_{L'_2})}{c(L,(l'_1,\dots,l'_{i-1},a,l'_{i+1},\dots,l'_r))} = \\
&\frac{(-1)^{z}}{\prod_{j=0}^{z-1}{(n-r-j)}} + (n - r - z)\frac{(-1)^{z+1}}{\prod_{j=0}^{z}{(n-r-j)}} = 0
\end{align*}
We can now prove the first statement of Theorem \ref{phipropertiestheorem}. Observe that
\begin{align*}
\sum_{\sigma \in S_{[1,n] \setminus I}}{\phi_{F,\sigma(L)}} &= \sum_{\sigma \in S_{[1,n] \setminus I}}{\sum_{L'}{c(\sigma(L),L')p_{F,L'}}} \\
&=\sum_{L'}{\sum_{\sigma \in S_{[1,n] \setminus I}}{c(L,\sigma^{-1}(L'))p_{F,L'}}}
\end{align*}
where the second equality follows from the fact that for all permutations $\sigma$, $c(L,L') = c(\sigma(L),\sigma(L'))$. Thus, it is sufficient to show that for all $L,L'$ and all $I \subsetneq I_L$, $\sum_{\sigma \in S_{[1,n] \setminus I}}{c(L,\sigma(L'))} = 0$. This can be shown as follows. Choose an index $i$ such that $l'_i \notin I$ (where $L' = (l'_1,\dots,l'_r)$) and observe that
\begin{align*}
&\sum_{\sigma \in S_{[1,n] \setminus I}}{c(L,\sigma(L'))} = \\
&\sum_{\sigma_2 \in S_{[1,n] \setminus (I \cup \{l'_i\})}}{\left(\sum_{a \in ([1,n] \setminus \sigma_2(I_{L'})) \cup \{l'_i\}}
{c(L,(\sigma_2(l'_1),\dots,\sigma_2(l'_{i-1}),a,\sigma_2(l'_{i+1}),\dots,\sigma_2(l'_{k}))))}\right)}
\end{align*}
which is $0$ by Lemma \ref{coefficientequationlemma}.
\end{proof}
We now prove the second statement of Theorem \ref{phipropertiestheorem}, which says that if $F_1$ and $F_2$ are flags such that $r_{F_1} \neq r_{F_2}$ then for any $\tilde{E}$ which is symmetric under permutations of $[1,n]$, 
for any tuples $L,L'$ of sizes $r_{F_1}$ and $r_{F_2}$, $\tilde{E}[\phi_{F_1,L} \cdot \phi_{F_2,L'}] = 0$.

Without loss of generality, assume that $r_{F_1} < r_{F_2}$. We prove this statement by induction on $|I_{L} \setminus I_{L'}|$. For the base case, if $I_{L} \subseteq I_{L'}$ then averaging over permutations of $[1,n] \setminus I_{L}$, $\phi_{F_1,L}$ is unchanged but by statement one of Theorem \ref{phipropertiestheorem}, $\phi_{F_2,L'}$ becomes $0$. Thus, by symmetry, $\tilde{E}[\phi_{F_1,L} \cdot \phi_{F_2,L'}] = 0$.
If $I_{L}$ is not a subset of $I_{L'}$ then there must be some $l_i$ which is in $I_{L}$ but not in $I_{L'}$. Now if $L = \{l_1,\dots,l_{r_{F_1}}\}$, applying statement one of Theorem \ref{phipropertiestheorem}, we obtain that 
\[
\sum_{a \in ([1,n] \setminus I_{L}) \cup \{l_i\})}{\phi_{F_1,\{l_1,\dots,l_{i-1},a,l_{i+1},\dots,l_{r_{F_1}}\}}} = 0
\]
This implies that 
\[
\sum_{a \in I_{L'} \setminus I_L}{\tilde{E}[\phi_{F_1,\{l_1,\dots,l_{i-1},a,l_{i+1},\dots,l_r\}} \cdot \phi_{F_2,L'}]} +\sum_{a \in ([1,n] \setminus (I_L \cup I_{L'})) \cup \{l_i\}}{\tilde{E}[\phi_{F_1,\{l_1,\dots,l_{i-1},a,l_{i+1},\dots,l_r\}} \cdot \phi_{F_2,L'}]} = 0
\]
By the inductive hypothesis, the terms in the left sum are all $0$. By symmetry, the right sum is equal to $(n - |I_L \cup I_{L'}| + 1)\tilde{E}[\phi_{F_1,L} \cdot \phi_{F_2,L'}]$. Thus, $\tilde{E}[\phi_{F_1,L} \cdot \phi_{F_2,L'}] = 0$, as needed.

Finally, we prove the third statement of Theorem \ref{phipropertiestheorem}, which says that for all polynomials $g$ such that $indexdeg(g) \leq \frac{n}{2}$, we can write $g = \sum_{F,L}{b_{F,L}\phi_{F,L}}$ where for all flags $F$, tuples of indices $L = (l_1,\dots,l_{r_F})$, and $i \in [1,r_F]$, 
\[
\sum_{a \in ([1,n] \setminus I_L) \cup l_i}{b_{F,\{l_1,\dots,l_{i-1},a,l_{i+1},\dots,l_{r_F}\}}} = 0
\]
To prove this statement, we do the following
\begin{enumerate}
\item Show that for any polynomial $g$ which is a linear combination of the polynomials $\{\phi_{F,L}\}$, we can write $g = \sum_{F,L}{b_{F,L}\phi_{F,L}}$ where for all $F$, $L = (l_1,\dots,l_{r_F})$, and $i \in [1,r_F]$, 
\[
\sum_{a \in ([1,n] \setminus I_L) \cup l_i}{b_{F,(l_1,\dots,l_{i-1},a,l_{i+1},\dots,l_{r_F})}} = 0
\]
\item Show that all polynomials $g$ are linear combinations of the polynomials $\{\phi_{F,L}\}$.
\end{enumerate}
For the first part, choose the coefficients $\{b_{F,L}\}$ to minimize the expression $\sum_{F,L}{b_{F,L}^2}$. We claim that for all $F$, $L = (l_1,\dots,l_{r_F})$, and $i \in [1,r_F]$, 
\[
\sum_{a \in ([1,n] \setminus I_L) \cup l_i}{b_{F,(l_1,\dots,l_{i-1},a,l_{i+1},\dots,l_{r_F})}} = 0
\]
To see this, note that for all $F$, $L = (l_1,\dots,l_{r_F})$, and $i \in [1,r_F]$, by the first statement of Theorem \ref{phipropertiestheorem} 
\[
\sum_{a \in ([1,n] \setminus I_L) \cup l_i}{\phi_{F,(l_1,\dots,l_{i-1},a,l_{i+1},\dots,l_{r_F})}} = 0
\]
Thus, we may shift all of the coefficients $\{b_{F,\{l_1,\dots,l_{i-1},a,l_{i+1},\dots,l_{r_F}\}}: a \in ([1,n] \setminus I_L) \cup l_i\}$ up or down by a constant without affecting 
$\sum_{F,L}{b_{F,L}\phi_{F,L}}$. If we had that 
\[
\sum_{a \in ([1,n] \setminus I_L) \cup l_i}{b_{F,(l_1,\dots,l_{i-1},a,l_{i+1},\dots,l_{r_F})}} \neq 0
\]
then we would be able to reduce $\sum_{F,L}{b_{F,L}^2}$, contradicting the minimality of $\sum_{F,L}{b_{F,L}^2}$.

For the second part, it is sufficient to show that for all $F,L$ such that $F$ has at most $\frac{n}{2}$ vertices, we can express $p_{F,L}$ as a linear combination of terms of the form $\phi_{F,L'}$ where $L'$ is a permutation of $L$ and terms of the form $p_{F',L'}$ where $F'$ has a smaller order than $F$. We defer this part of the proof to subsection \ref{decomposingpflsubsection} because it requires ideas from the proof of Theorem \ref{symmetrysquaredecompositiontheorem} which is given in the next subsection.
\end{proof}
\subsection{Explicit decomposition of $\tilde{E}[g^2]$}
We are now ready to decompose $\tilde{E}[g^2]$.
\begin{theorem}\label{symmetrysquaredecompositiontheorem}
If $g = \sum_{F,L}{b_{F,L}\phi_{F,L}}$ has index degree at most $\frac{n}{2}$ and for all $F$, $L = (l_1,\dots,l_{r_F})$, and $i \in [1,r_F]$, 
\[
\sum_{a \in ([1,n] \setminus I_L) \cup l_i}{b_{F,(l_1,\dots,l_{i-1},a,l_{i+1},\dots,l_{r_F})}} = 0
\]
then
\begin{align*}
&\tilde{E}[g^2] = \tilde{E}\left[\sum_{F,F': r_{F'} = r_F = r}{\sum_{L,L'}{b_{F,L}\phi_{F,L}b_{F',L'}\phi_{F',L'}}}\right] =\\
&\sum_{A \subseteq [1,r]}{\sum_{L:\forall i < i' \notin A, l_i < l_{i'},\forall i < i' \in A, l_i < l_{i'}}{}}\Big(\\ 
&\frac{\tilde{E}\left[\left(\sum_{F:r_F = r}{\sum_{\pi_0 \in S_{[1,r] \setminus A}}{\left(\sum_{\pi \in S_{r}: \forall i \notin A, \pi(i) = \pi_0(i)}{b_{F,(l_{\pi(1)},\dots,l_{\pi(r)})}}\right)
\left(\sum_{\pi \in S_{r}: \forall i \notin A, \pi(i) = \pi_0(i)}{\phi_{F,(l_{\pi(1)},\dots,l_{\pi(r)})}}\right)}}\right)^2\right]}{|A|!(n-r)(n-r-1)\dots(n-r-|A|+1)}\Big)
\end{align*}
\end{theorem}
\begin{proof}
To analyze $\tilde{E}[g^2] = \sum_{F,L,F',L'}{\tilde{E}[b_{F,L}\phi_{F,L}b_{F',L'}\phi_{F',L'}]}$, we take each $F,L$ and analyze 
\[
\sum_{F',L'}{\tilde{E}[b_{F,L}\phi_{F,L}b_{F',L'}\phi_{F',L'}]}
\]
Letting $r = r_{F}$, by the second statement of Theorem \ref{phipropertiestheorem}, we only need to consider $F'$ of order $r$. 

For $r = 0$, 
\[
\sum_{F':r_{F'} = 0}{\tilde{E}[b_{F,\emptyset}\phi_{F,\emptyset}b_{F',\emptyset}\phi_{F',\emptyset}]} = \tilde{E}\left[b_{F,\emptyset}\phi_{F,\emptyset}\left(\sum_{F':r_{F'} = 0}{b_{F',\emptyset}\phi_{F',\emptyset}}\right)\right]
\]
Thus, 
\[
\sum_{F,F':r_{F'} = r_F = 0}{\tilde{E}[b_{F,\emptyset}\phi_{F,\emptyset}b_{F',\emptyset}\phi_{F',\emptyset}]} = \tilde{E}\left[\left(\sum_{F:r_F = 0}{b_{F,\emptyset}\phi_{F,\emptyset}}\right)^2\right]
\]

For $r = 1$, for each $F$ and $i \in [1,n]$ we have that
\begin{align*}
&\sum_{F',i':r_{F'} = 1}{\tilde{E}[b_{F,(i)}\phi_{F,(i)}b_{F',(i')}\phi_{F',(i')}]} = \\
&\tilde{E}\left[b_{F,(i)}\phi_{F,(i)}\sum_{F':r_{F'} = 1}\left(b_{F',(i)}\phi_{F',(i)} + \sum_{i' \neq i}{b_{F',(i')}\phi_{F',(i')}}\right)\right] 
\end{align*}
The key idea is that by using the properties of the functions $\{\phi_{F',(i')}\}$ and the coefficients $\{b_{F',(i')}\}$, we can transform terms with an index $i' \neq i$ into terms which only have the index $i$. By the first statement of Theorem \ref{phipropertiestheorem}, for all $F'$ of asymmetry level $1$, $\sum_{i' \neq i}{\phi_{F',(i')}} = -\phi_{F',(i)}$. By symmetry, this implies that for all flags $F'$ of asymmetry level $1$, 
\[
\tilde{E}[\phi_{F,(i)}\phi_{F',(i')}] = -\frac{1}{n-1}\tilde{E}[\phi_{F,(i)}\phi_{F',(i)}]
\]
Since $\sum_{i'}{b_{F',(i')}} = 0$ for all flags $F'$ of asymmetry level $1$, 
\[
\tilde{E}\left[\phi_{F,(i)}\sum_{i' \neq i}{b_{F',(i')}\phi_{F',(i')}}\right] = \tilde{E}\left[-\phi_{F,(i)}\frac{\sum_{i' \neq i}{b_{F',(i')}}}{n-1}\phi_{F',(i)}\right] = 
\tilde{E}\left[\frac{b_{F',(i)}\phi_{F,(i)}\phi_{F',(i)}}{n-1}\right]
\]
Substituting this equation into the equation above,
\begin{align*}
&\sum_{F', i':r_{F'} = 1}{\tilde{E}[b_{F,(i)}\phi_{F,(i)}b_{F',(i')}\phi_{F',(i')}]} = \\
&\left(1 + \frac{1}{n-1}\right)\tilde{E}\left[b_{F,(i)}\phi_{F,(i)}\left(\sum_{F':r_{F'} = 1}{b_{F',(i)}\phi_{F',(i)}}\right)\right]
\end{align*}
Summing this equation over all $F,i$, our final result for $r = 1$ is 
\begin{align*}
&\sum_{F,F',i,i':r_{F'} = r_F = 1}{\tilde{E}[b_{F,(i)}\phi_{F,(i)}b_{F',(i')}\phi_{F',(i')}]} = \\
&\frac{n}{n-1}\sum_{i=1}^{n}{\tilde{E}\left[\left(\sum_{F:r_F = 1}{b_{F,(i)}\phi_{F,(i)}}\right)^2\right]}
\end{align*}

We can use similar ideas for $r=2$, though it is somewhat more complicated. For each $F,i,j$ such that $i \neq j$ we have that
\begin{align*}
&\sum_{F',i',j':j' \neq i', r_{F'} = 2}{\tilde{E}[b_{F,(i,j)}\phi_{F,(i,j)}b_{F',(i',j')}\phi_{F',(i',j')}]} = \\
&\sum_{F': r_{F'} = 2}\tilde{E}\Big[b_{F,(i,j)}\phi_{F,(i,j)}\Big(b_{F',(i,j)}\phi_{F',(i,j)} + b_{F',(j,i)}\phi_{F',(j,i)} + \\
&\sum_{j' \notin (i,j)}{\left(b_{F',(i,j')}\phi_{F',(i,j')} + b_{F',(j',i)}\phi_{F',(j',i)}\right)} + \sum_{i' \notin (i,j)}{\left(b_{F',(i',j)}\phi_{F',(i',j)} + b_{F',(j,i')}\phi_{F',(j,i')}\right)} + \\
&\sum_{i' \notin (i,j), j' \notin \{i,j,i'\}}{b_{F',(i',j')}\phi_{F',(i',j')}}\Big)\Big]
\end{align*}
Similar to the analysis for $r=1$, we use the properties of the functions $\{\phi_{F',(i',j')}\}$ and the coefficients $\{b_{F',(i',j')}\}$ to transform terms where $i' \notin \{i,j\}$ or $j' \notin \{i,j\}$ into terms which only have the indices $i,j$. The results are as follows:
\begin{lemma}\label{requalstwolemma} \ 
\begin{enumerate}
\item $\sum_{j' \notin (i,j)}{\tilde{E}[\phi_{F,(i,j)}b_{F',(i,j')}\phi_{F',(i,j')}]} = \frac{1}{n-2}\tilde{E}[\phi_{F,(i,j)}b_{F',(i,j)}\phi_{F',(i,j)}]$
\item $\sum_{j' \notin (i,j)}{\tilde{E}[\phi_{F,(i,j)}b_{F',(j',i)}\phi_{F',(j',i)}]} = \frac{1}{n-2}\tilde{E}[\phi_{F,(i,j)}b_{F',(j,i)}\phi_{F',(j,i)}]$
\item $\sum_{i' \notin (i,j)}{\tilde{E}[\phi_{F,(i,j)}b_{F',(i',j)}\phi_{F',(i',j)}]} = \frac{1}{n-2}\tilde{E}[\phi_{F,(i,j)}b_{F',(i,j)}\phi_{F',(i,j)}]$
\item $\sum_{i' \notin (i,j)}{\tilde{E}[\phi_{F,(i,j)}b_{F',(i,j')}\phi_{F',(i,j')}]} = \frac{1}{n-2}\tilde{E}[\phi_{F,(i,j)}b_{F',(j,i)}\phi_{F',(j,i)}]$
\item 
\[
\sum_{i' \notin (i,j), j' \notin \{i,j,i'\}}{\tilde{E}\left[\phi_{F,(i,j)}b_{F',(i',j')}\phi_{F',(i',j')}\right]} = 
\tilde{E}\left[\frac{b_{F',(i,j)} + b_{F',(j,i)}}{(n-2)(n-3)}\phi_{F,(i,j)}(\phi_{F',(i,j)} + \phi_{F',(j,i)})\right]
\]
\end{enumerate}
\end{lemma}
\begin{proof}
By the first statement of Theorem \ref{phipropertiestheorem}, for all $F'$ of asymmetry level $2$, $\sum_{j' \neq i,j}{\phi_{F',(i,j')}} = -\phi_{F',(i,j)}$. By symmetry, 
\[
\forall j' \notin \{i,j\}, \tilde{E}[\phi_{F,(i,j)}\phi_{F',(i,j')}] = \frac{1}{n-2}\tilde{E}[\phi_{F,(i,j)}\phi_{F',(i,j)}]
\]
Since $\sum_{j' \neq i,j}{b_{F',(i,j')}} = -b_{F',(i,j)}$, 
\begin{align*}
\sum_{j' \notin (i,j)}{\tilde{E}[\phi_{F,(i,j)}b_{F',(i,j')}\phi_{F',(i,j')}]} &= -\frac{1}{n-2}\left(\sum_{j' \notin (i,j)}{b_{F',(i,j')}}\right)\tilde{E}[\phi_{F,(i,j)}\phi_{F',(i,j)}] \\
&= \frac{1}{n-2}\tilde{E}[\phi_{F,(i,j)}b_{F',(i,j)}\phi_{F',(i,j)}]
\end{align*}
The second, third, and fourth statements can be proved using similar logic. For the final statement, by statement one of Theorem \ref{phipropertiestheorem}, for all $F'$ of asymmetry level $2$, 
\[
\sum_{i' \notin (i,j), j' \notin \{i,j,i'\}}{\phi_{F',(i',j')}} = -\sum_{i' \neq i,j}{(\phi_{F',\{i',i\}} + \phi_{F',(i',j)})}
= \phi_{F',(j,i)} + \phi_{F',(i,j)}
\]
By symmetry, this implies that 
\[
\forall i' \notin (i,j), \forall j' \notin \{i,j,i'\}, \tilde{E}[\phi_{F,(i,j)}\phi_{F',(i',j')}] = \frac{1}{(n-2)(n-3)}\tilde{E}[\phi_{F,(i,j)}(\phi_{F',(i,j)} + \phi_{F',(j,i)})]
\]
We also have that
\[
\sum_{i' \notin (i,j), j' \notin \{i,j,i'\}}{b_{F',(i',j')}} = -\sum_{i' \neq i,j}{(b_{F',\{i',i\}} + b_{F',(i',j)})} = b_{F',(j,i)} + b_{F',(i,j)}
\]
Combining these equations, 
\begin{align*}
\sum_{i' \notin (i,j), j' \notin \{i,j,i'\}}{\tilde{E}\left[\phi_{F,(i,j)}b_{F',(i',j')}\phi_{F',(i',j')}\right]} &= 
\left(\sum_{i' \notin (i,j), j' \notin \{i,j,i'\}}{b_{F',(i',j')}}\right)\tilde{E}\left[\phi_{F,(i,j)}\phi_{F',(i',j')}\right]\\
&=\tilde{E}\left[\frac{b_{F',(i,j)} + b_{F',(j,i)}}{(n-2)(n-3)}\phi_{F,(i,j)}(\phi_{F',(i,j)} + \phi_{F',(j,i)})
\right]
\end{align*}
as needed
\end{proof}
Plugging the equations of Lemma \ref{requalstwolemma} into the expression for 
\[\sum_{F',i',j':j' \neq i', r_{F'} = 2}{\tilde{E}[b_{F,(i,j)}\phi_{F,(i,j)}b_{F',(i',j')}\phi_{F',(i',j')}]}
\]
we have that
\begin{align*}
&\sum_{F',i',j':j' \neq i', r_{F'} = 2}{\tilde{E}[b_{F,(i,j)}\phi_{F,(i,j)}b_{F',(i',j')}\phi_{F',(i',j')}]} = \\
&\sum_{F': r_{F'} = 2}\tilde{E}\Big[b_{F,(i,j)}\phi_{F,(i,j)}\Big(b_{F',(i,j)}\phi_{F',(i,j)} + b_{F',(j,i)}\phi_{F',(j,i)} + \\
&\sum_{j' \notin (i,j)}{\left(b_{F',(i,j')}\phi_{F',(i,j')} + b_{F',(j',i)}\phi_{F',(j',i)}\right)} + \sum_{i' \notin (i,j)}{\left(b_{F',(i',j)}\phi_{F',(i',j)} + b_{F',(j,i')}\phi_{F',(j,i')}\right)} + \\
&\sum_{i' \notin (i,j), j' \notin \{i,j,i'\}}{b_{F',(i',j')}\phi_{F',(i',j')}}\Big)\Big] = \\
&\sum_{F': r_{F'} = 2}\tilde{E}\Big[b_{F,(i,j)}\phi_{F,(i,j)}\Big(b_{F',(i,j)}\phi_{F',(i,j)} + b_{F',(j,i)}\phi_{F',(j,i)} + \\
&\frac{2}{n-2}b_{F',(i,j)}\phi_{F',(i,j)} + \frac{2}{n-2}b_{F',(j,i)}\phi_{F',(j,i)} + 
\frac{b_{F',(i,j)} + b_{F',(j,i)}}{(n-2)(n-3)}(\phi_{F',(i,j)} + \phi_{F',(j,i)})\Big)\Big] =\\
&\sum_{F': r_{F'} = 2}\tilde{E}\Big[b_{F,(i,j)}\phi_{F,(i,j)}\Big(\frac{n}{n-2}b_{F',(i,j)}\phi_{F',(i,j)} + \frac{n}{n-2}b_{F',(j,i)}\phi_{F',(j,i)} + \\
&\frac{b_{F',(i,j)} + b_{F',(j,i)}}{(n-2)(n-3)}(\phi_{F',(i,j)} + \phi_{F',(j,i)})\Big)\Big]
\end{align*}
Summing this equation over all $F,i,j$ where $j \neq i$,
\begin{align*}
&\sum_{F,i,j,F',i',j':j \neq i, j' \neq i', r_F = r_{F'} = 2}{\tilde{E}[b_{F,(i,j)}\phi_{F,(i,j)}b_{F',(i',j')}\phi_{F',(i',j')}]} = \\
&\sum_{F,i,j,F',i',j':i < j, j' \neq i', r_F = r_{F'} = 2}{\tilde{E}[\left(b_{F,(i,j)}\phi_{F,(i,j)} + b_{F,(j,i)}\phi_{F,(j,i)}\right)b_{F',(i',j')}\phi_{F',(i',j')}]} = \\
&\sum_{i,j: i < j}{\tilde{E}\Big[\sum_{F:r_F = 2}{\left(b_{F,(i,j)}\phi_{F,(i,j)} + b_{F,(j,i)}\phi_{F,(j,i)}\right)}} \\
&\Big(\frac{n}{n-2}\sum_{F':r_{F'} = 2}{\left(b_{F',(i,j)}\phi_{F',(i,j)} + b_{F',(j,i)}\phi_{F',(j,i)}\right)} + \\
&\sum_{F':r_{F'} = 2}{\frac{b_{F',(i,j)} + b_{F',(j,i)}}{(n-2)(n-3)}(\phi_{F',(i,j)} + \phi_{F',(j,i)})}\Big)\Big] = \\
&\sum_{i,j: i < j}{\tilde{E}\left[\left(\sum_{F:r_F = 2}{\left(b_{F,(i,j)}\phi_{F,(i,j)} + b_{F,(j,i)}\phi_{F,(j,i)}\right)}\right)^2\right]} + \\
&\frac{1}{(n-2)(n-3)}\sum_{i,j: i < j}{\tilde{E}\Big[\left(\sum_{F:r_F = 2}{\left(b_{F,(i,j)}\phi_{F,(i,j)} + b_{F,(j,i)}\phi_{F,(j,i)}\right)}\right)} \\
&{\left(\sum_{F':r_{F'} = 2}{(b_{F',(i,j)} + b_{F',(j,i)})(\phi_{F',(i,j)} + \phi_{F',(j,i)})}\right)\Big]}
\end{align*}
Note that there is a mismatch between 
\[
\sum_{F:r_F = 2}{\left(b_{F,(i,j)}\phi_{F,(i,j)} + b_{F,(j,i)}\phi_{F,(j,i)}\right)}
\]
and 
\[
\sum_{F':r_{F'} = 2}{(b_{F',(i,j)} + b_{F',(j,i)})(\phi_{F',(i,j)} + \phi_{F',(j,i)})}
\]
To handle this, we use symmetry with respect to swapping $i$ and $j$. By this symmetry, 
\begin{align*}
&\tilde{E}[\phi_{F,(i,j)}(\phi_{F',(i,j)} + \phi_{F',(j,i)})] = \tilde{E}[\phi_{F,(j,i)}(\phi_{F',(i,j)} + \phi_{F',(j,i)})] \\
&= \frac{1}{2}\tilde{E}[(\phi_{F,(i,j)} + \phi_{F,(j,i)})(\phi_{F',(i,j)} + \phi_{F',(j,i)})]
\end{align*}
This implies that for all $i < j$, 
\begin{align*}
&\tilde{E}\Big[\left(\sum_{F:r_F = 2}{\left(b_{F,(i,j)}\phi_{F,(i,j)} + b_{F,(j,i)}\phi_{F,(j,i)}\right)}\right) \\
&\left(\sum_{F':r_{F'} = 2}{(b_{F',(i,j)} + b_{F',(j,i)})(\phi_{F',(i,j)} + \phi_{F',(j,i)})}\right)\Big] = \\
&\tilde{E}\Big[\left(\sum_{F:r_F = 2}{(b_{F,(i,j)} + b_{F,(j,i)})\frac{\phi_{F,(i,j)} + \phi_{F,(j,i)}}{2}}\right) \\
&\left(\sum_{F':r_{F'} = 2}{(b_{F',(i,j)} + b_{F',(j,i)})(\phi_{F',(i,j)} + \phi_{F',(j,i)})}\right)\Big]
\end{align*}
Our final result for $r = 2$ is 
\begin{align*}
&\sum_{F,i,j,F',i',j':j \neq i, j' \neq i', r_F = r_{F'} = 2}{\tilde{E}[b_{F,(i,j)}\phi_{F,(i,j)}b_{F',(i',j')}\phi_{F',(i',j')}]} = \\
&\sum_{i,j:i<j}{\tilde{E}\left[\frac{n}{n-2}\left(\sum_{F:r_F = 2}{\left(b_{F,(i,j)}\phi_{F,(i,j)} + b_{F,(j,i)}\phi_{F,(j,i)}\right)}\right)^2\right]} + \\
&\sum_{i,j:i<j}{\tilde{E}\left[\frac{1}{2(n-2)(n-3)}\left(\sum_{F:r_F = 2}{\left(b_{F,(i,j)} + b_{F,(j,i)}\right)\left(\phi_{F,(i,j)} + \phi_{F,(j,i)}\right)}\right)^2\right]}
\end{align*}
To see the general pattern, we prove the following lemmas and corollaries.
\begin{lemma}\label{generalblemma}
Let $L$ be an ordered set of size $r$, let $A$ be a subset of $[1,r]$, and let $\pi_0$ be a permutation of $[1,r] \setminus A$. For all flags $F'$ of order $r$, 
\[
\sum_{j_1,j_2,\dots,j_r:\forall i \notin A, j_i = l_{\pi_0(i)}, \{j_i: i \in A\} \subseteq [1,n] \setminus I_L, \atop \{j_i: i \in A\} \text{ are all distinct}, 
}{b_{F',({j_1},\dots,{j_r})}} = (-1)^{|A|}\sum_{\pi \in S_{r}: \forall i \notin A, \pi(i) = \pi_0(i)}{b_{F',(l_{\pi(1)},\dots,l_{\pi(r)})}}
\]
\end{lemma}
\begin{proof}
Choose an index $i \in A$. Using the equation $\sum_{j_i \notin \{j_1,\dots,j_{i-1},j_{i+1},\dots,j_r\}}{b_{F',({j_1},\dots,{j_r})}} = 0$, we obtain that for all 
$j_1,\dots,j_{i-1},j_{i+1},\dots,j_r$, 
\[
\sum_{j_i \notin I_L \cup \{j_1,\dots,j_{i-1},j_{i+1},\dots,j_r\}}{b_{F',({j_1},\dots,{j_r})}} = 
-\sum_{j': j' \in I_L \setminus  \{j_1,\dots,j_{i-1},j_{i+1},\dots,j_r\}}{b_{F',(j_1,\dots,j_{i-1},j',j_{i+1},\dots,j_r)}}
\]
This replaces the sum over $j_i$ by an index in $I_L \setminus \{j_1,\dots,j_r\}$. Applying this logic repeatedly, the result follows.
\end{proof}
\begin{lemma}
Let $L$ be an ordered set of size $r$, let $A$ be a subset of $[1,r]$, and let $\pi_0$ be a permutation of $[1,r] \setminus A$. For all flags $F'$ of order $r$, 
\[
\sum_{j_1,j_2,\dots,j_r:\forall i \notin A, j_i = l_{\pi_0(i)}, \{j_i: i \in A\} \subseteq [1,n] \setminus I_L, \atop \{j_i: i \in A\} \text{ are all distinct}, 
}{\phi_{F',({j_1},\dots,{j_r})}} = (-1)^{|A|}\sum_{\pi \in S_{r}: \forall i \notin A, \pi(i) = \pi_0(i)}{\phi_{F',(l_{\pi(1)},\dots,l_{\pi(r)})}}
\]
\end{lemma}
\begin{proof}
This can be proved in exactly the same way.
\end{proof}
\begin{corollary}
Let $L$ be an ordered set of size $r$, let $A$ be a subset of $[1,r]$, and let $\pi_0$ be a permutation of $[1,r] \setminus A$. For any $\tilde{E}$ which is symmetric, any flags $F,F'$ of order $r$, and any $j_1,\dots,j_r$ such that $j_i = l_{\pi_0(i)}$ whenever $i \notin A$, $j_i \in [1,n] \setminus I_L$ whenever $i \in A$, and $\{j_i: i \in A\}$ are all distinct, 
\[
\tilde{E}[\phi_{F,L}\phi_{F',\{j_1,j_2,\dots,j_r\}}] = 
\frac{(-1)^{|A|}\tilde{E}\left[\phi_{F,L}\left(\sum_{\pi \in S_{r}: \forall i \notin A, \pi(i) = \pi_0(i)}{\phi_{F',\{l_{\pi(1)},\dots,l_{\pi(r)}\}}}\right)\right]}{(n-r)(n-r-1)\dots(n-r-|A|+1)}
\]
\end{corollary}
\begin{proof}
This follows from symmetry and the fact that
\begin{align*}
&\sum_{j_1,j_2,\dots,j_r:\forall i \notin A, j_i = l_{\pi_0(i)}, \{j_i: i \in A\} \subseteq [1,n] \setminus I_L, \atop \{j_i: i \in A\} \text{ are all distinct}, 
}{\tilde{E}[\phi_{F,L}\phi_{F',\{j_1,j_2,\dots,j_r\}}]} = \\
&(-1)^{|A|}\tilde{E}\left[\phi_{F,L}\left(\sum_{\pi \in S_{r}: \forall i \notin A, \pi(i) = \pi_0(i)}{\phi_{F',\{l_{\pi(1)},\dots,l_{\pi(r)}\}}}\right)\right]
\end{align*}
\end{proof}
\begin{corollary}
Let $L$ be an ordered set of size $r$, let $A$ be a subset of $[1,r]$, and let $\pi_0$ be a permutation of $[1,r] \setminus A$. For any $\tilde{E}$ which is symmetric and any flags $F,F'$ of order $r$,
\begin{align*}
&\sum_{j_1,j_2,\dots,j_r:\forall i \notin A, j_i = l_{\pi_0(i)}, \{j_i: i \in A\} \subseteq [1,n] \setminus I_L, \atop \{j_i: i \in A\} \text{ are all distinct}, 
}{\tilde{E}[\phi_{F,L}b_{F',\{{j_1},\dots,{j_r}\}}\phi_{F',\{j_1,j_2,\dots,j_r\}}]} = \\
&\frac{\tilde{E}\left[\phi_{F,L}\left(\sum_{\pi \in S_{r}: \forall i \notin A, \pi(i) = \pi_0(i)}{b_{F',\{l_{\pi(1)},\dots,l_{\pi(r)}\}}}\right)\left(\sum_{\pi \in S_{r}: \forall i \notin A, \pi(i) = \pi_0(i)}{\phi_{F',\{l_{\pi(1)},\dots,l_{\pi(r)}\}}}\right)\right]}{(n-r)(n-r-1)\dots(n-r-|A|+1)}
\end{align*}
\end{corollary}
\begin{corollary}
Let $L$ be an ordered set of size $r$ and let $A$ be a subset of $[1,r]$. For any $\tilde{E}$ which is symmetric and any flags $F,F'$ of order $r$,
\begin{align*}
&\sum_{j_1,j_2,\dots,j_r:\exists \pi_0 \in S_{[1,r] \setminus A}: \forall i \notin A, j_i = l_{\pi_0(i)}, \{j_i: i \in A\} \subseteq [1,n] \setminus I_L, \atop \{j_i: i \in A\} \text{ are all distinct}, 
}{\tilde{E}[\phi_{F,L}b_{F',\{{j_1},\dots,{j_r}\}}\phi_{F',\{j_1,j_2,\dots,j_r\}}]} = \\
&\frac{\tilde{E}\left[\phi_{F,L}\sum_{\pi_0 \in S_{[1,r] \setminus A}}\left(\sum_{\pi \in S_{r}: \forall i \notin A, \pi(i) = \pi_0(i)}{b_{F',\{l_{\pi(1)},\dots,l_{\pi(r)}\}}}\right)\left(\sum_{\pi \in S_{r}: \forall i \notin A, \pi(i) = \pi_0(i)}{\phi_{F',\{l_{\pi(1)},\dots,l_{\pi(r)}\}}}\right)\right]}{(n-r)(n-r-1)\dots(n-r-|A|+1)}
\end{align*}
\end{corollary}
To avoid a mismatch, we use the following proposition:
\begin{proposition}
Let $L$ be an ordered set of size $r$, let $A$ be a subset of $[1,r]$, and let $\pi_0$ be a permutation of $[1,r] \setminus A$. For any $\tilde{E}$ which is symmetric and any flags $F,F'$ of order $r$, for all permutations $\pi' \in S_r$ such that $\pi'(i) = i$ whenever $i \notin A$, 
\begin{align*}
&\tilde{E}\left[\phi_{F,\{l_{\pi'(1)},\dots,l_{\pi'(r)}\}}\left(\sum_{\pi \in S_{r}: \forall i \notin A, \pi(i) = \pi_0(i)}{\phi_{F',\{l_{\pi(1)},\dots,l_{\pi(r)}\}}}\right)\right] = \\
&\frac{1}{|A|!}\tilde{E}\left[\left(\sum_{\pi' \in S_{r}: \forall i \notin A, \pi'(i) = i}{\phi_{F,\{l_{\pi'(1)},\dots,l_{\pi'(r)}\}}}\right)\left(\sum_{\pi \in S_{r}: \forall i \notin A, \pi(i) = \pi_0(i)}{\phi_{F',\{l_{\pi(1)},\dots,l_{\pi(r)}\}}}\right)\right]
\end{align*}
\end{proposition}
\begin{proof}
This can be shown by using symmetry with respect to permutations of $\{l_i: i \in A\}$.
\end{proof}
Putting everything together, we obtain the following corollary
\begin{corollary}
Let $A$ be a subset of $[1,r]$. For any $\tilde{E}$ which is symmetric,
\begin{align*}
&\tilde{E}\left[\sum_{F,F': r_{F'} = r_F = r}{\sum_{L,L':\exists \pi_0 \in S_{[1,r] \setminus A}: \forall i \notin A, l'_i = l_{\pi_0(i)}, \{l'_i:i \in A\} \cap \{l_i: i \in A\} = \emptyset}{b_{F,L}\phi_{F,L}b_{F',L'}\phi_{F',L'}}}\right] =\\
&\sum_{L:\forall i < i' \notin A, l_i < l_{i'},\forall i < i' \in A, l_i < l_{i'}}{\Big(}\\ 
&\frac{\tilde{E}\left[\left(\sum_{F:r_F = r}{\sum_{\pi_0 \in S_{[1,r] \setminus A}}{\left(\sum_{\pi \in S_{r}: \forall i \notin A, \pi(i) = \pi_0(i)}{b_{F,\{l_{\pi(1)},\dots,l_{\pi(r)}\}}}\right)
\left(\sum_{\pi \in S_{r}: \forall i \notin A, \pi(i) = \pi_0(i)}{\phi_{F,\{l_{\pi(1)},\dots,l_{\pi(r)}\}}}\right)}}\right)^2\right]}{|A|!(n-r)(n-r-1)\dots(n-r-|A|+1)}\Big)
\end{align*}
\end{corollary}
Summing this equation over all possible $A \subseteq [1,r]$ we obtain that
\begin{align*}
&\tilde{E}\left[\sum_{F,F': r_{F'} = r_F = r}{\sum_{L,L'}{b_{F,L}\phi_{F,L}b_{F',L'}\phi_{F',L'}}}\right] =\\
&\sum_{A \subseteq [1,r]}{\sum_{L:\forall i < i' \notin A, l_i < l_{i'},\forall i < i' \in A, l_i < l_{i'}}{}}\Big(\\ 
&\frac{\tilde{E}\left[\left(\sum_{F:r_F = r}{\sum_{\pi_0 \in S_{[1,r] \setminus A}}{\left(\sum_{\pi \in S_{r}: \forall i \notin A, \pi(i) = \pi_0(i)}{b_{F,\{l_{\pi(1)},\dots,l_{\pi(r)}\}}}\right)
\left(\sum_{\pi \in S_{r}: \forall i \notin A, \pi(i) = \pi_0(i)}{\phi_{F,\{l_{\pi(1)},\dots,l_{\pi(r)}\}}}\right)}}\right)^2\right]}{|A|!(n-r)(n-r-1)\dots(n-r-|A|+1)}\Big)
\end{align*}
as needed.
\end{proof}

\subsection{Decomposition of $p_{F,L}$}\label{decomposingpflsubsection}
In this subsection, we confirm that $p_{F,L}$ can be expressed as a linear combination of the functions $\{\phi_{F',L'}\}$ where $F'$ is obtained from $F$ by making some of the labeled vertices unlabeled.
\begin{theorem}\label{pfldecompositiontheorem}
For any flag $F$ such that $r = r_F \leq \frac{n}{2}$, there exist coefficients $c(F,L,F',L')$ such that 
\[
p_{F,L} = \sum_{F',L'}{c(F,L,F',L')\phi_{F',L'}}
\]
where $c(F,L,F',L')$ is only nonzero for $F'$ which are obtained by taking labeled vertices in $F$ and making them unlabeled
\end{theorem}
\begin{proof}
To prove this theorem, we show the following lemma which expresses $\phi_{F,L}$ in terms of the functions $\{p_{F,L'}: I_{L'} = I_L\}$ up to lower order terms. This gives us a matrix $M$ which changes basis from the functions $\{\phi_{F,L}\}$ to the functions $\{p_{F,L}\}$ up to lower order terms. We then show directly that $M$ is invertible, which implies that we can express $p_{F,L}$ in terms of the functions $\{\phi_{F,L'}: I_{L'} = I_L\}$ up to lower order terms. This allows us to prove the theorem by induction.
\begin{lemma}\label{phiflreexpressionlemma}
Let $L$ be an ordered set of size $r \in [0,\frac{n}{2}]$. For all flags $F$ of order $r$,
\[
\phi_{F,L} = \sum_{A \subseteq [1,r]}{\frac{\sum_{\pi \in S_{r}: \forall i \notin A, \pi(i) = i}{p_{F,\{l_{\pi(1)},\dots,l_{\pi(r)}\}}}}{\prod_{j=0}^{|A|-1}{(n-r-j)}}} + 
\text{ lower order terms}
\]
where the lower order terms are all of the form $c_{F',L'}p_{F',L'}$ where $F'$ is obtained from $F$ by taking some of the labeled vertices and making them unlabeled.
\end{lemma}
Before proving this lemma, we show how it implies Theorem \ref{pfldecompositiontheorem} using the strategy described above.
\begin{definition}
Define $M$ to be the matrix indexed by permutations $\pi \in S_{r}$ with entries
\[
M_{{\pi}{\pi'}} = \sum_{A: \{i: \pi(i) \neq \pi'(i)\} \subseteq A  \subseteq [1,r]}{\frac{1}{\prod_{j=0}^{|A|-1}{(n-r-j)}}}
\]
\end{definition}
\begin{proposition}
Let $r$ be an integer in $[0,\frac{n}{2}]$, let $L = (l_1,\dots,l_r)$ be an ordered set of size $r$, and let $F$ be a flag of order $r$. Given a vector $c$ indexed by permutations $\pi \in S_{r}$,
\[
\sum_{\pi \in S_r}{c_{\pi}\phi_{F,\{l_{\pi(1)},\dots,l_{\pi(r)}\}}} = \sum_{\pi' \in S_r}{(Mc)_{\pi'}p_{F,\{l_{\pi'(1)},\dots,l_{\pi'(r)}\}}}  + 
\text{ lower order terms}
\]
\end{proposition}
\begin{proof}
Using Lemma \ref{phiflreexpressionlemma}, letting $L^{\pi} = (l^{\pi}_1,\dots,l^{\pi}_r)$ be the ordered set where $l^{\pi}_{i} = l_{\pi(i)}$, 
\begin{align*}
\sum_{\pi \in S_r}{c_{\pi}\phi_{F,\{l_{\pi(1)},\dots,l_{\pi(r)}\}}} &= \sum_{\pi \in S_r}{c_{\pi}\phi_{F,L^{\pi}}} = \sum_{\pi \in S_r}{c_{\pi}\phi_{F,\{l^{\pi}_{1},\dots,l^{\pi}_{r}\}}} \\
&= \sum_{\pi \in S_r}{c_{\pi}\sum_{A \subseteq [1,r]}{\frac{\sum_{\pi_2 \in S_{r}: \forall i \notin A, \pi_2(i) = i}{p_{F,\{l^{\pi}_{\pi_2(1)},\dots,l^{\pi}_{\pi_2(r)}\}}}}{\prod_{j=0}^{|A|-1}{(n-r-j)}}}} + \text{ lower order terms}\\
&= \sum_{\pi \in S_r}{c_{\pi}\sum_{A \subseteq [1,r]}{\frac{\sum_{\pi_2 \in S_{r}: \forall i \notin A, \pi_2(i) = i}{p_{F,\{l_{\pi{\pi_2}(1)},\dots,l_{\pi{\pi_2}(r)}\}}}}{\prod_{j=0}^{|A|-1}{(n-r-j)}}}} + \text{ lower order terms}
\end{align*}
Taking $\pi' = {\pi}\pi_2$, we have that 
\begin{align*}
\sum_{\pi \in S_r}{c_{\pi}\phi_{F,\{l_{\pi(1)},\dots,l_{\pi(r)}\}}} &= \sum_{\pi \in S_r}{c_{\pi}\sum_{A \subseteq [1,r]}{\frac{\sum_{\pi' \in S_{r}: \forall i \notin A, \pi'(i) = \pi(i)}{p_{F,\{l_{\pi'(1)},\dots,l_{\pi'(r)}\}}}}{\prod_{j=0}^{|A|-1}{(n-r-j)}}}} + \text{ lower order terms}\\
&= \sum_{\pi,\pi' \in S_r}{\sum_{A: \{i: \pi'(i) \neq \pi(i)\} \subseteq A  \subseteq [1,r]}{\frac{c_{\pi}p_{F,\{l_{\pi'(1)},\dots,l_{\pi'(r)}\}}}{\prod_{j=0}^{|A|-1}{(n-r-j)}}}} + \text{ lower order terms}\\
&= \sum_{\pi' \in S_r}{(Mc)_{\pi'}p_{F,\{l_{\pi'(1)},\dots,l_{\pi'(r)}\}}}  + \text{ lower order terms}
\end{align*}
\end{proof}
\begin{lemma}
$M \succ 0$
\end{lemma}
\begin{proof}
We can decompose $M$ as 
\[
M = \sum_{\pi_0 \in S_r}{\sum_{A \subseteq [1,r]}{\frac{1_{\pi_0,A}1^{T}_{\pi_0,A}}{|A|!\prod_{j=0}^{|A|-1}{(n-r-j)}}}}
\]
where $1_{\pi_0,A}(\pi) = 1$ if $\pi(i) = \pi_0(i)$ for all $i \notin A$ and $1_{\pi_0,A}(\pi) = 0$ otherwise. To check this, we compare the entries of these matrices
\begin{align*}
&\left(\sum_{\pi_0 \in S_r}{\sum_{A \subseteq [1,r]}{\frac{1_{\pi_0,A}1^{T}_{\pi_0,A}}{|A|!\prod_{j=0}^{|A|-1}{(n-r-j)}}}}\right)_{{\pi}{\pi'}} = \\
&\sum_{\pi_0 \in S_r}{\sum_{A: \{i: \pi(i) \neq \pi_0(i)\} \cup \{i: \pi'(i) \neq \pi_0(i)\} \subseteq A \subseteq [1,r]}{\frac{1}{|A|!\prod_{j=0}^{|A|-1}{(n-r-j)}}}} = \\
&\sum_{A: \{i: \pi(i) \neq \pi'(i)\} \subseteq A \subseteq [1,r]}{\sum_{\pi_0 \in S_r: \{i:\pi_0(i) \neq \pi(i)\} \subseteq A}{\frac{1}{|A|!\prod_{j=0}^{|A|-1}{(n-r-j)}}}} = \\
&\sum_{A: \{i: \pi(i) \neq \pi'(i)\} \subseteq A \subseteq [1,r]}{\frac{1}{\prod_{j=0}^{|A|-1}{(n-r-j)}}} = M_{{\pi}{\pi'}}
\end{align*}
We now observe that 
\[
M \succeq \sum_{\pi_0 \in S_r}{1_{\pi_0,\emptyset}1^{T}_{\pi_0,\emptyset}} = Id \succ 0
\]
which completes the proof.
\end{proof}
\begin{corollary}\label{phiflreexpressioncorollary}
Let $L$ be an ordered set of size $r \in [0,\frac{n}{2}]$. For all flags $F$ of asymmetry level $r$ and tuples $L$ of size $r$, given a vector $c$ indexed by permutations $\pi \in S_{r}$, 
\[
\sum_{\pi \in S_r}{c_{\pi}p_{F,\{l_{\pi(1)},\dots,l_{\pi(r)}\}}} = \sum_{\pi' \in S_r}{(M^{-1}c)_{\pi'}\phi_{F,\{l_{\pi'(1)},\dots,l_{\pi'(r)}\}}}  + 
\text{ lower order terms}
\]
where the lower order terms are all of the form $c_{F',L'}p_{F',L'}$ where $F'$ is obtained from $F$ by taking some of the labeled vertices and making them unlabeled.
\end{corollary}
Theorem \ref{pfldecompositiontheorem} follows easily from Corollary \ref{phiflreexpressioncorollary}. Taking $c$ to be the vector indexed by permutations $\pi \in S_r$ such that $c_{\pi} = 1$ if $\pi = Id$ and $c_{\pi}$, 
\[
p_{F,L} = \sum_{\pi' \in S_r}{(M^{-1}c)_{\pi'}\phi_{F,\{l_{\pi'(1)},\dots,l_{\pi'(r)}\}}} + \text{ lower order terms}
\]
where the lower order terms are all of the form $c_{F',L'}p_{F',L'}$ where $F'$ is obtained from $F$ by taking some of the labeled vertices and making them unlabeled.
\begin{proof}[Proof of Lemma \ref{phiflreexpressionlemma}]
Recall that 
\[
\phi_{F,L} = \sum_{L'}{c(L,L')p_{F,L'}}
\]
where $c(L,L')$ is $0$ if $\exists i,j: i \neq j, l'_j = l_i$ and
\[
c(L,L') = \frac{(-1)^{|\{i:l'_i \neq l_i\}|}}{\prod_{j=0}^{|\{i:l'_i \neq l_i\}|-1}{(n-r-j)}}
\]
otherwise. 
\[
\phi_{F,L} = \sum_{A \subseteq [1,r]}{\sum_{L' = (l'_1,\dots,l'_r): \forall i \notin A, l'_i = l_i, \forall i \in A, l'_i \notin I_L}{\frac{(-1)^{|A|}p_{F,L'}}{\prod_{j=0}^{|A|-1}{(n-r-j)}}}}
\]
We analyze this expression by considering each possible subset $A \subseteq [1,r]$ separately.
\begin{lemma}
Let $L$ be a tuple of size $r$ and let $A$ be a subset of $[1,r]$. For all flags $F$ of order $r$, 
\[
\sum_{L' = (l'_1,\dots,l'_r): \forall i \notin A, l'_i = l_i, \forall i \in A, l'_i \notin I_L}{p_{F,L'}} = (-1)^{|A|}\sum_{\pi \in S_{r}: \forall i \notin A, \pi(i) = i}{p_{F,\{l_{\pi(1)},\dots,l_{\pi(r)}\}}} + \text{ lower order terms}
\]
where the lower order terms are all of the form $c_{F',L'}p_{F',L'}$ where $F'$ is obtained from $F$ by taking some of the labeled vertices and making them unlabeled.
\end{lemma}
\begin{proof}
Recall that for all $i \in [1,r]$, for all ordered sets $L$ of size $r$,
\[
\sum_{a \in ([1,n] \setminus I_L) \cup \{l_i\}}{p_{F,\{l_1,\dots,l_{i-1},a,l_{i+1},\dots,l_{r}\}}} = p_{F',L'}
\]
where $F'$ is the flag obtained by taking $F$ and making the vertex $v_i$ unlabeled and $L' = (l_1,\dots,l_{i-1},l_{i+1},\dots,l_{r})$ is the tuple obtained by deleting $l_i$ from $L$. Thus, 
\[
\sum_{a \in ([1,n] \setminus I_L) \cup \{l_i\}}{p_{F,\{l_1,\dots,l_{i-1},a,l_{i+1},\dots,l_{r}\}}} = 0 + \text{ lower order term}
\]
where the lower order term is $p_{F',L'}$ where $F'$ is obtained from $F$ by taking a labeled vertex of $F$ and making it unlabled

We now use the same logic as we used to prove Lemma \ref{generalblemma}. Choose an index $i \in A$. Using the equation 
\[
\sum_{j_i \notin \{j_1,\dots,j_{i-1},j_{i+1},\dots,j_r\}}{p_{F,\{{j_1},\dots,{j_r}\}}} = 0  + \text{ lower order terms}
\]
we obtain that for all $j_1,\dots,j_{i-1},j_{i+1},\dots,j_r$, 
\[
\sum_{j_i \notin I_L \cup \{j_1,\dots,j_{i-1},j_{i+1},\dots,j_r\}}{p_{F,\{{j_1},\dots,{j_r}\}}} = 
-\sum_{j': j' \in I_L \setminus  \{j_1,\dots,j_{i-1},j_{i+1},\dots,j_r\}}{p_{F,\{j_1,\dots,j_{i-1},j',j_{i+1},\dots,j_r\}}}
\]
This replaces the sum over $j_i$ by an index in $I_L \setminus \{j_1,\dots,j_r\}$. Applying this logic repeatedly, the result follows.
\end{proof}
Lemma \ref{phiflreexpressionlemma} now follows by summing this lemma over all possible $A$.
\end{proof}
\end{proof}
\subsection{Proof of Theorem \ref{squarereductiontheorem}}
We now complete the proof of Theorem \ref{squarereductiontheorem}. We have to show that if $\tilde{E}$ is a linear map from polynomials of index degree at most $d \leq n$ to $\mathbb{R}$ which is symmetric with respect to permutations of $[1,n]$ then for any polynomial $g$ of index degree $d' \leq \frac{d}{2}$, we can write
\[
\tilde{E}[g^2] = \sum_{I \subseteq [1,n],j:|I|\leq d'}{\tilde{E}[g^2_{Ij}]}
\]
where for all $I,j$,
\begin{enumerate}
\item $g_{Ij}$ is symmetric with respect to permutations of $[1,n] \setminus I$.
\item $indexdeg(g_{Ij}) \leq indexdeg(g)$
\item $\forall i \in I, \sum_{\sigma \in S_{[1,n] \setminus (I \setminus \{i\})}}{\sigma(g_{Ij})} = 0$
\end{enumerate}
By the third statement of Theorem \ref{phipropertiestheorem}, for all polynomials $g$ such that $indexdeg(g) \leq \frac{n}{2}$, we can write $g = \sum_{F,L}{b_{F,L}\phi_{F,L}}$ where $b_{F,L} = 0$ whenever $r_{F} > indexdeg(g)$ and for all flags $F$, tuples of indices $L = (l_1,\dots,l_{r_F})$, and $i \in [1,r_F]$, 
\[
\sum_{a \in ([1,n] \setminus I_L) \cup l_i}{b_{F,\{l_1,\dots,l_{i-1},a,l_{i+1},\dots,l_{r_F}\}}} = 0
\]
Theorem \ref{symmetrysquaredecompositiontheorem} implies that we can write 
\[
\tilde{E}[g^2] = \sum_{I \subseteq [1,n],j:|I|\leq indexdeg(g)}{\tilde{E}[g^2_{Ij}]}
\]
where each $g_{Ij}$ is a linear combination of terms of the form $\phi_{F,L}$ where $r_F = |I|$ and $I_L = I$. We now verify the three required statements. Given the form of $g_{Ij}$, the first and second statements are trivial. For the third statement, for all $i \in I$, letting $j$ be the index such that $l_j = i$, 
\[
\sum_{a: a \in ([1,n] \setminus I) \cup \{i\}}{\phi_{F,\{l_1,\dots,l_{j-1},a,l_{j+1},\dots,l_{|I|}\}}} = 0
\]
By symmetry and linearity, we must have that
\[
\forall i \in I, \sum_{\sigma \in S_{[1,n] \setminus (I \setminus \{i\})}}{\sigma(g_{Ij})} = 0
\]
as needed.
\end{appendix}
\end{document}